\documentclass[11pt]{amsart}
\usepackage{amssymb,amsmath,amsthm}
\usepackage{a4wide}
\usepackage{graphicx}

\usepackage{color}
\usepackage{mathrsfs}
\usepackage{latexsym}
\usepackage{bbm}
\usepackage{dsfont}

\numberwithin{equation}{section}

\theoremstyle{plain}
\begingroup
\newtheorem{theorem}{Theorem}[section]
\newtheorem{lemma}[theorem]{Lemma}
\newtheorem{proposition}[theorem]{Proposition}

\endgroup

\theoremstyle{definition}
\begingroup
\newtheorem{definition}[theorem]{Definition}
\newtheorem{remark}[theorem]{Remark}

\endgroup

\usepackage{graphicx}

\newcommand{\nsc}{\mathrm{nsc}}

\newcommand{\dg}{\mathrm{eu}}

\newcommand{\Indi}{\mathbbm{1}}
\newcommand{\AC}{\mathcal{AC}}
\newcommand{\Fa}{\mathsf{F}}

\newcommand{\gr}{\mathsf{gr}}

\newcommand{\F}{\mathcal{F}}

\newcommand{\M}{\mathcal M}

\newcommand{\N}{\mathbb{N}}
\newcommand{\Z}{\mathbb{Z}}
\newcommand{\R}{\mathbb{R}}
\newcommand{\clos}{\mathrm{clos}}

\newcommand{\C}{\mathcal{C}}

\newcommand{\defe}{\mathrm{def}}

\newcommand{\E}{\mathcal{E}}

\newcommand{\supp}{\mathrm{supp}\;}

\newcommand{\ud}{\,\mathrm{d}}

\renewcommand\tilde{\widetilde}

\newcommand{\ii}{\bar\imath}

\newcommand{\res}{\mathop{\hbox{\vrule height 7pt width .5pt depth 0pt
\vrule height .5pt width 6pt depth 0pt}}\nolimits}

\newcommand{\T}{\mathcal T}

\newcommand{\weakly}{\rightharpoonup}
\newcommand{\weakstar}{\overset{\ast}{\rightharpoonup}}
\newcommand{\Fsq}{\mathsf{F}^{\diamondsuit}}
\newcommand{\defl}{\mathrel{\mathop:}=}
\newcommand{\bad}{\mathrm{bad}} 
\newcommand{\Add}{\mathsf{Add}} 
\newcommand{\G}{\mathsf{G}}
\newcommand{\Fbad}{\mathsf{F}^{\mathrm{bad}}}
\newcommand{\eu}{\mathrm{Euler}}

\def\Om{\Omega}

\def\E{\mathcal{E}}

\def\F{\mathcal{F}}

\def\dist{\textup{dist}}

\def\1{\mathbf{1}}
\def\loc{\mathrm{loc}}
\def\Int{\mathrm{int}}

\def\XXint#1#2#3{{\setbox0=\hbox{$#1{#2#3}{\int}$ }
\vcenter{\hbox{$#2#3$ }}\kern-.57\wd0}}

\def\E{\mathcal{E}}

\newcommand{\I}{\mathcal{I}}
\newcommand{\Ed}{\mathsf{Ed}}
\newcommand{\Wi}{\mathrm{wire}}
\newcommand{\ext}{\mathrm{ext}}
\newcommand{\bow}{\mathrm{bow}}
\newcommand{\inte}{\mathrm{int}}
\newcommand{\X}{\mathsf{X}}
\newcommand{\Y}{\mathsf{Y}}
\newcommand{\Per}{\mathrm{Per}}

\def\red#1{\textcolor{red}{#1}}

\title[Vectorial crystallization] {Vectorial crystallization problems  \\
and collective behavior}

\author[L. De Luca]
{L. De Luca}
\address[Lucia De Luca]{
IAC-CNR, Via dei Taurini, 19 I-00184 Rome, Italy
}
\email[L. De Luca]{lucia.deluca@cnr.it}

\author[A. Ninno]
{A. Ninno}
\address[Angelo Ninno]{
Dipartimento di Matematica ``G. Castelnuovo'', Sapienza Universit\`a di Roma, Piazzale A. Moro 2, I-00185, Rome, Italy
}
\email[A. Ninno]{angelo.ninno@uniroma1.it}

\author[M. Ponsiglione]
{M. Ponsiglione}
\address[Marcello Ponsiglione]{Dipartimento di Matematica ``G. Castelnuovo'', Sapienza Universit\`a di Roma, Piazzale A. Moro 2, I-00185, Rome, Italy
}
\email[M. Ponsiglione]{ponsigli@mat.uniroma1.it}
\begin{document}
 \maketitle

%

\begin{abstract}
	We propose and analyze a class of vectorial crystallization problems, with applications to crystallization of anisotropic molecules and collective behavior such as birds flocking and fish schooling. 
	
	We focus on  two-dimensional systems of ``oriented'' particles: Admissible configurations are represented by vectorial empirical measures with density in $\mathcal S^1$. We endow such configurations with a graph structure, where the bonds represent
	the ``convenient'' interactions between particles, and 
	the proposed variational principle consists in maximizing their  number. 
	The class of bonds is determined by hard sphere type  pairwise potentials, depending both on the distance between the particles and on the angles between the segment joining two particles and their orientations, through threshold criteria. 
		
	Different ground states emerge by tuning the angular dependence in the potential, mimicking ducklings swimming  in a row formation and predicting as well,  for some specific values of the angular parameter,  the so-called {\it diamond formation} in fish schooling. 
	\vskip5pt
	\noindent
	\textsc{Keywords: Crystallization; Collective behavior; Graph theory; Variational methods.} 
	\vskip5pt
	\noindent
	\textsc{AMS subject classifications:} 70C20, 05C10,  49J45,  82D25.
\end{abstract}
\tableofcontents
\section*{Introduction}
The crystallization problem consists in understanding periodic configurations of atoms or molecules;   self organization in ordered configurations is a  general issue in nature and a central problem in biology \cite{PaKe}: birds flocking and fish schooling are typical examples of collective behavior and formation of ordered structures \cite{Rey, Fish}. 

Variational principles are transversal in nature, and in view of their predictive properties and ability to synthesize modeling perspectives, they have been a sound and fruitful  reading key of many complex phenomena; relevant configurations of particles in ordered structures can be detected as ground states of energy functionals:
Classical potentials adopted in the variational formulations involve  pairwise interactions depending on the mutual distances between the particles; among them, short range repulsive/long range
 attractive potentials are very relevant, predicting in particular  regular triangular lattices \cite{Theil} (see also \cite{BDP} for a crystallization result in the square lattice).

In this paper, we consider systems composed by elements, referred to as particles, endowed with an orientation that affects their  interactions, and in turn the specific ground states. This is the case of anisotropic, e.g., elongated molecules (as proteins often are), where the orientation is given by the shape of the molecule, as well as the case of collective behaviors, where each individual is oriented according with its visual cone, or (mostly equivalently) with its velocity. Such orientations may  induce highly anisotropic structures: Patchy particles with tetrahedral symmetry form diamond lattices \cite{NVDL} and diamond formations turn out to be energetically efficient for swimming
of schooled fishes \cite{Bre65,W73,JCLiao};  ducklings move in a row formation, while the so called {\it V-formation} is convenient for ducks as well as for many other birds migration \cite{Fish}.  

Although this kind of interactions, depending on mutual positions and orientations, attracted much attention in modeling and simulation issues, their mathematical analysis received until now only occasional interest.  Also the theory of {\it boids}  \cite{Rey},  developed  to replicate birds flight,  while having a  big impact also  for  entertainment purposes, to the best of our knowledge, still lacks of  a rigorous theoretical analysis.      

The purpose of this paper is to provide a solid mathematical ground to describe and attack this kind of problems, and to  propose and analyze in details a  basic simple model in two dimensions able to predict diamond formations.  The model is based on classical variational approaches to  crystallization problems, but dealing with vectorial (rather than scalar) empirical measures with densities on the unit circle $\mathcal S^1$, taking into account the orientation of the particles. 
The optimal configuration is assumed to maximize the number of convenient interactions between particles, where the notion of convenience is determined by the mutual position and orientation between pairs of particles. 
Loosely speaking, our model is based on representing the admissible configurations as graphs, where the bonds are determined by the vectorial empirical measures, while the variational principle consists just in maximizing the number of such (convenient) bonds.  A relevant issue  of the proposed  model is that, in defining the graph structure, and, in turn, the class of bonds, the positional and orientational variables are coupled; in this respect our model is also able to predict, rather than assume, alignment of orientations. Clearly, the specific rules determining the favorable interactions should rely on specific modeling considerations: for instance, in a pair of  caudal swimmers the energy spent by the follower depends on the vortices in the fluid produced by the leader, and hence its energetically saving position depends on the angle between the direction of swimming and that of the line joining the two swimmers \cite{Fish, Bre76}.  The terrific simplification we introduce in these modeling issues consists in assuming that both positional and orientational variables determine the class of convenient interactions only through threshold criteria, tuned by only two parameters, one acting on the mutual positions and the other one on both mutual positions and orientations. The positional short range repulsive/long range attractive behavior is mimicked by the hard sphere formalism as in the Heitmann-Radin model \cite{HeRa} (see also \cite{AFS, DLF1, DNP, KP, FKS}): particles are represented as discs with fixed radius that cannot overlap and that may interact only if they are tangent each other. Such interaction is switched on only if the angle formed by the orientations of both particles and the segment joining them is smaller than a given angular parameter, acting as a threshold. Therefore, renormalizing the hard disc radius, the model depends only on the angular parameter representing the amplitude of the visual cone of each particle.

Varying such a parameter, we find different ground states (see Subsection \ref{examples}). Clearly,  the $0$-angular visual cone enforces one-dimensional ground states where particles have all the same orientation and are aligned on a row oriented accordingly, mimicking cyclist and duckling formations.
On the other hand, tuning off the dependence on the orientation variable, the proposed model gives back the Heitmann-Radin model and the corresponding regular triangular lattice. Remarkably, for a specific choice of the angular visual cone, we show that diamond configurations minimize the energy; furthermore, in the thermodynamic limit as the number of particles diverges we show that all configurations with suitable energy bounds from above, 
consist in patched configurations of diamond formations with bounded perimeter.

While  the specific model analyzed in details is very  basic,  the proposed methodology seems to be robust and 
 could be enriched and generalized in many directions.
As for instance, one could relax the threshold criteria seeking for other relevant configurations, like the V-formation. 
Further restrictions of our model are the two-dimensional setting and the fact that long range interactions are neglected.
Moreover, our purely variational and static approach could  be compared with reacher models \cite{ucc} accounting for 
topological rather than metric distances,
fluctuations and dynamical aspects.
In particular, we remark that the orientational variable in our model has only a purely ideal interpretation in terms of the velocity of the particles; a justification of such an interpretation and rigorous connections with true dynamical models could deserve further investigation. 

The paper is organized as follows.
In Section \ref{prelimgra} we collect some preliminaries on planar graphs that will be used throughout the paper.
Section \ref{ilmodello} is devoted to the description of our model and to the analysis of the qualitative properties of the ground states.
In Section \ref{rad32sec} we prove the minimality of the diamond formation and the compactness properties of quasi-minimizers of the energy. 
Finally, some technical lemmas and their proofs are collected in Appendix \ref{aux2}.
\section{Preliminaries on planar graphs}\label{prelimgra}

Here we collect some notions and notation on planar graphs that will be adopted in this paper.

Let $\X$ be a finite subset of $\R^2$ and let $\Ed$ be a given subset of $\mathsf{E}(\X)$\,, where
\begin{equation}\label{tuttibond}
\mathsf{E}(\X):= \{\{x,y\}\subset \R^2\,:\, x,y\in\X\,,\,x\neq y\}\,.
\end{equation}

 The pair $\G=(\X,\Ed)$ is called {\it graph};  $\X$ is called  the set of {\it vertices} of $\G$ and $\Ed$ is called the set of {\it edges} (or {\it bonds}) of $\G$\,.

Given $\X'\subset\X$ we denote by $\G_{\X'}$ the {\it subgraph} (or {\it restriction}) of $\G$ generated by $\X'$\,, defined by $\G_{\X'}=(\X',\Ed')$ where $\Ed':=\{\{x',y'\}\in\Ed\,:\, x',y'\in\X'\}$\,.

\begin{definition}\label{conncomp}
We say that two points $x,z\in\X$ are connected and we write $x\sim z$ if there exist $M\in\N$ and a {\it path} $x=y_0,\ldots,y_M=z$ such that $\{y_{m-1},y_m\}\in\Ed$ for every $m=1,\ldots,M-1$\,. We say that  $\G_{\X_1},\ldots,\G_{\X_K}$ with $K\in\N$ are the {\it connected components} of $\G$ if 
$\{\X_1,\ldots,\X_{K}\}$ is a partition of $\X$
and for every $k,k'\in\{1,\ldots,K\}$ with $k\neq k'$ it holds
\begin{align*}
x_k\sim y_k\qquad&\textrm{for every }x_k,y_k\in\X_k\,,\\
x_{k}\not\sim x_{k'}\qquad&\textrm{for every }x_k\in\X_k\,, x_{k'}\in\X_{k'}\,.
\end{align*}
If $\G$ has only one connected component we say that $\G$ is {\it connected}.
\end{definition}

We say that $\G$ is planar if for every pair of (distinct) bonds $\{x_1,x_2\}, \{y_1,y_2\}\in\Ed$, the (open) segments $(x_1,x_2)$ and $(y_1,y_2)$ have empty intersection.

From now on we assume that $\G=(\X,\Ed)$ is planar, so that we can introduce the notion of face (see also \cite{DLF1}).

By a face $f$ of $\G$ we mean any open, bounded, connected component of 
$
\R^2\setminus \big(\X\cup\bigcup_{\{x,y\}\in\Ed}[x,y]\big)$, which is also simply connected; here $[x,y]$ is the closed segment with extreme points $x$ and $y$. We denote by $\Fa^{\mathrm{sc}}(\G)$\,, or simply by $\Fa(\G)$\,, the set of faces of $\G$\,. Moreover, we denote by $\Fa^{\nsc}(\G)$ the set of open, bounded, connected components of 
$
\R^2\setminus \big(\X\cup\bigcup_{\{x,y\}\in\Ed}[x,y]\big)$, which are not simply connected.
We warn the reader that, in standard literature, also the elements of $\Fa^{\nsc}(\G)$ are called faces.
Moreover we set
\begin{equation*}
O(\G):=\bigcup_{f\in\Fa(\G)}\clos(f)\,.
\end{equation*}

With a little abuse of language we will say that an edge $\{x,y\}$ lies on a set $E\subset\R^2$ if the segment $[x,y]$ is contained in $E$\,.
We classify the edges in $\Ed$ in the following subclasses:
\begin{itemize}
\item $\Ed^{\Int}$ is the set of {\it interior edges}, i.e., of edges lying on the boundary of two (distinct) faces; 
\item $\Ed^{\Wi,\ext}$ is the set of {\it exterior wire edges}, i.e., of edges that do not lie on the boundary of any face;
\item $\Ed^{\Wi,\mathrm{int}}$ is the set of {\it interior wire edges}, i.e., of edges lying on the boundary of precisely one face but not on the boundary of its closure (or, equivalently, of $O(\G)$)\,;
\item $\Ed^{\partial}$  is the set of {\it boundary edges}, i.e., of edges lying on $\partial O(\G)$\,. 
\end{itemize}
%


Analogously, for every face $f\in \Fa(\G)$ one can define the following subclasses of edges delimiting $f$:
\begin{itemize}
\item $\Ed^{\Wi,\inte}(f)$ is the set of edges lying on the boundary of $f$ but not on the boundary of the closure of $f$;
\item $\Ed^{\partial}(f)$ is the set of edges lying on the boundary of the closure of $f$.
\end{itemize}
Finally, we define the Euler characteristic of the graph $\G=(\X,\Ed)$. To this purpose,
set $l_0(\G)=l^{\eu}_0(\G):=\sharp \X$\,, $l_1(\G)=l^{\eu}_1(\G):=\sharp \Ed$\,, $l_2(\G):=\sharp \Fa(\G)$\,, and $l^{\eu}_2(\G):=l_2(\G)+\sharp \Fa^{\nsc}(\G)$.  Then, we introduce the standard Euler characteristic $\chi^{\eu}(\G)$ of $\G$\,, together with a slight variant $\chi(\G)$ of it that will be useful for our purposes:
\begin{equation} \label{Euchar}
\chi^{\eu}(\G) := \sum_{k=0}^2 (-1)^k l^{\eu}_k(\G)\,,\qquad \chi(\G) := \sum_{k=0}^2 (-1)^k l_k(\G)\,.
\end{equation}
In the next result we recall the classical Euler characteristic formula for planar graphs together with its analogous for $\chi(\G)$\,.
\begin{lemma}\label{lm:bouchar}
For every planar graph $\G$ the Euler characteristic $\chi^{\eu}(\G)$ is equal to the number of connected components of $\G$\,. 
Moreover
\begin{equation}\label{bouchar}
\chi(\G)\ge 1\,,
\end{equation}
and the equality in \eqref{bouchar} holds true whenever $\G$ is connected.
\end{lemma}
\begin{proof}
The first sentence in the statement is the classical Euler formula. 
As for the proof of formula \eqref{bouchar}, we observe that $\chi^{\eu}(\G)-\chi(\G)=\sharp\Fa^{\nsc}(\G)$ and that
 for every $f\in\Fa^{\nsc}(\G)$\,, each ``hole'' of $f$ ``contains'' at least a connected component of $\G$\,; by an easy induction argument on the number $\sharp\Fa^{\nsc}(\G)$ of non simply connected faces one can get \eqref{bouchar}.

Finally, if $\G$ is connected, then $\sharp\Fa^{\nsc}(\G)=0$ so that $\chi(\G)=\chi^{\eu}(\G)=1$\,.
\end{proof}

We define the {\it graph-perimeter} of $\G$ as
\begin{equation}\label{graphperG}
\Per_{\gr}(\G):=\sharp\Ed^{\partial}+2\sharp\Ed^{\Wi,\ext}\,.
\end{equation}

Analogously, the {\it graph-perimeter} of a face $f$ is defined by
 \begin{equation}\label{graphgeoper}
\Per_{\gr}(f):=\sharp\Ed^{\partial}(f)+2\sharp \Ed^{\Wi,\inte}(f).
\end{equation}

Now we show how a planar graph can be triangulated, controlling the number of the required additional bonds.
Since the triangulation procedure is local, we will focus on graphs having only one face.

\begin{lemma}\label{triang}
Let $\G=(\X,\Ed)$ be a planar graph having only one face $f$. If $\Per_{\gr}(f)\ge 4$, then there exists a planar graph $\bar\G=(\X,\bar\Ed)$ such that
\begin{itemize}
\item[(1)] $\Ed\subset\bar\Ed$;
\item[(2)] all the faces of $\Fa(\bar\G)$ are triangles;
\item[(3)] $(x,y)\subset f$ for every $\{x,y\}\in\bar\Ed\setminus\Ed$;
\item[(4)] $\sharp\bar\Ed=\sharp\Ed+\Per_{\gr}(f)-3$;
\item[(5)] $\sharp\Fa(\bar\G)=\Per_{\gr}(f)-2$. 
\end{itemize}
\end{lemma}

\begin{proof}
We proceed by induction on $\Per_{\gr}(f)$. 
If $\Per_{\gr}(f)=3,4$, the claims are clearly true.  Assume now that $\Per_{\gr}(f)\ge 5$ and that the claims are satisfied for every graph whose only face $g$ satisfies $\Per_{\gr}(g)<\Per_{\gr}(f)$.
We can always split $f$ into two simply connected sets $f_1,f_2$ by cutting $f$ with a segment $[x,y]$ joining two points $x,y\in \X$ such that $(x,y)$ is contained in $f$ 
(see for instance \cite{ears}).
Setting $\Ed':=\Ed\cup\{x,y\}$, we have that the faces of $\G'=(\X,\Ed')$ are exactly $f_1$ and $f_2$.
Moreover, we have that
\begin{equation}\label{si}
\Per_{\gr}(f_1)+\Per_{\gr}(f_2)=\Per_{\gr}(f)+2.
\end{equation}
As a consequence, $\Per_{\gr}(f_i)<\Per_{\gr}(f)$ for $i=1,2$. 
By induction, we can triangulate both $f_1$ and $f_2$ adding $\Per_{\gr}(f_1)-3$ and $\Per_{\gr}(f_2)-3$ edges respectively. In view of \eqref{si} this implies that the number of additional edges used to triangulate $f$ is 
$$
1+\Per_{\gr}(f_1)-3+\Per_{\gr}(f_2)-3=\Per_{\gr}(f)-3.
$$ 
Moreover, again by the inductive step and by \eqref{si} we have that
$$
\sharp\Fa(\bar\G)=\Per_{\gr}(f_1)-2+\Per_{\gr}(f_2)-2=\Per_{\gr}(f)-2.
$$
\end{proof}

\section{The variational model}\label{ilmodello}
In this section we describe our model, introducing the energy functional, and we study qualitative properties of its minimizers and almost minimizers.
\subsection{Vectorial empirical configurations and their energy}
Let $\AC$ be a given subset of 
\begin{equation*}
\C:=\{(X,V)\,:\, X=(x_1,\ldots,x_N)\in(\R^2)^N,\, V=(v_1,\ldots,v_N)\in(\mathcal{S}^1)^N, \,N\in\N \}\,, 
\end{equation*}
where $\mathcal{S}^1$ denotes the set of unitary vectors of $\R^2$\,. Here and throughout the paper $\N$ denotes the set of positive integers.
In the following we will refer to $\AC$ as the set of {\it admissible (vectorial empirical) configurations}.

For every $(X,V)\in\C$ with $X=(x_1,\ldots,x_N)$ we set $\X:=\{x_1,\ldots,x_N\}$; we will adopt such a notation also for a generic subset of $\R^2$ with $N$ elements.
Moreover if $(X,V)\in\AC_N$\,, for every $x\in\X$ we denote by $v(x)$ the orientation associated to the point $x$\,, i.e., if $x=x_i$ for some $i=1,\ldots,N$\,, then $v(x)=v_i$\,.

We assume that there exists a given map $\Ed$ that at each element $(X,V)$ of $\AC$ associates a subset $\Ed(X,V)$ of $\mathsf{E}(\X)$ (defined in \eqref{tuttibond}), referred to as the set  {\it edges} (or {\it bonds}) of $(X,V)$. 

We call {\it energy} any given functional $\E:\AC\to\R$ of the form $\E(X,V)=\mathscr{E}(X,V,\Ed(X,V))$;
in this paper we consider energies of the following type
\begin{equation}\label{defE0}
\E(X,V):=-\sharp \Ed(X,V).
\end{equation}
Such a choice of energy is very specific and of course it could be generalized in many perspectives; nevertheless, it leaves enough  freedom in the choice of the criteria that determine the class $\AC$ of admissible configurations and the class of edges.

In this paper we focus on {\it threshold criteria} for defining both $\AC$ and $\Ed$. 
We set
\begin{equation*}
\AC:=\{(X,V)\in\C\,:\, |x-y|\ge 1\textrm{ for every }x,y\in\X\textrm{ with }x\neq y\}\,.
\end{equation*}

Given $\gamma\in[0,1]$, for every $(X,V)\in\C$ with $X=(x_1,\ldots,x_N)\in(\R^2)^N$ and $V=(v_1,\ldots,v_N)\in(\mathcal{S}^1)^N$ for some $N\in\N$ we define
\begin{equation}\label{edgegamma}
\begin{aligned}
\Ed^{\gamma}(X,V):=&\left\{\{x_i,x_j\}\in \mathsf{E}(\X)\,:\, |x_i-x_j|=1, \textrm{ and } \right.\\
&\left.\textrm{ either }\langle x_j-x_i,v_k\rangle \ge \gamma \textrm{ for }k=i,j
 \textrm{ or }  \langle x_i-x_j,v_k\rangle \ge \gamma\textrm{ for }k=i,j\right\}
\end{aligned}
\end{equation}
and we set
\begin{equation}\label{energygamma}
\E^\gamma(X,V):=-\sharp\Ed^{\gamma}(X,V).
\end{equation}
In Section \ref{rad32sec} we will focus on the case $\gamma=\frac{\sqrt{3}}{2}$; in Subsection \ref{examples} below, we briefly discuss what happens for different values of $\gamma$. The main purpose is to show the variety of ground states emerging from the only minimization of the basic energy $\E^\gamma$,  depending on the choice of the edges in \eqref{edgegamma}, i.e., on $\gamma$. 



\subsection{Discrete graph representation} \label{bonddef}

Fix $\gamma\in[0,1]$.
For every $(X,V)\in\AC$, we consider the graph $\G^\gamma(X,V)=(\X,\Ed^\gamma(X,V))$,
referred to as the {\it bond graph} of the configuration $(X,V)$.
Since $(X,V)\in\AC$, we have that $\G^\gamma(X,V)$ is planar.

From now we will use the notions introduced in Section \ref{prelimgra} for $\G=\G^\gamma(X,V)$ and $\Ed=\Ed^\gamma(X,V)$. To ease the notation we set $\Fa^\gamma(X,V):=\Fa(\G^\gamma(X,V))$\,, $O^\gamma(X,V):=O(\G^\gamma(X,V))$ and so on. Moreover, we denote by $\Ed^{\gamma,\partial}(X,V)$ the set of boundary edges of $\G^\gamma(X,V)$.

We define the perimeter $\Per^\gamma(X,V)$ of $\G^\gamma(X,V)$ as the perimeter of $O^\gamma(X,V)$, i.e.,
\begin{equation}\label{perimeter}
\Per^\gamma(X,V)\defl\sharp \Ed^{\gamma,\partial}(X,V).
\end{equation}
Notice that, since the edges have unitary length, $\Per^\gamma(X,V)$ is the standard perimeter of $O^\gamma(X,V)$, whereas $\Per^\gamma_{\gr}(X,V):=\Per_{\gr}(\G^\gamma(X,V))$ represents the relaxation of $\Per^\gamma(X,V)$ with respect to outer approximations of $O^\gamma(X,V)$ with open sets.

%
%
%
%

\begin{definition}\label{borX0}
For every $(X,V)\in\AC$ we denote by $\partial^\gamma X$ the subset of $\X$ made by points which lie either on (at least) one boundary edge or on (at least) one exterior wire edge or on no edge.
In the following, we will call {\it interior points} of $(X,V)$ the elements of $\X\setminus\partial ^\gamma X$.
\end{definition}



\subsection{Basic qualitative properties of the ground states of $\E^\gamma$}\label{examples}


We first recall that the Heitmann-Radin energy defined in \cite{HeRa} is given by the number of pairs of tangent hard spheres. Such a model fits our formalism: the energy can be written as in \eqref{defE0} where the bonds are determined by the only hard sphere tangency condition, being independent of the orientations.    
More precisely, 
for every $(X,V)\in\C$ the Heitmann-Radin energy can be defined by 
$$
\E_{HR}(X,V)=\left\{\begin{array}{ll}
-\sharp\Ed_{HR}(X)&\textrm{if }(X,V)\in\AC\\
+\infty&\textrm{otherwise}
\end{array}\right.\,,
$$
with 
$$
\Ed_{HR}(X):=\left\{ \{x,y \}\in \mathsf{E}(\X)\ :\ |x-y|=1 \right\}.
$$
As proven in \cite{HeRa} (see also \cite{DLF1}), for every fixed $N\in\N$ the minimizers of the functional $\E_{HR}$ in 
\begin{equation*}
\begin{aligned}
\AC_N&:=\{(X,V)\in\AC\cap \Big((\R^2)^N\times(\mathcal{S}^1)^N\Big)\}\,.
\end{aligned}
\end{equation*}
are, up to a rotation and a translation, subsets of the unitary triangular lattice
\begin{equation}\label{T1}
\T^1:=\left\{a(1;0)+b\Big(\frac 1 2;\frac{\sqrt 3}{2}\Big)\,:\,a,b\in\Z\right\}.
\end{equation}
For every $\gamma\in[0,1]$ and  for every $(X,V)\in\AC$,
we clearly have 
$\sharp \Ed^{\gamma}(X,V)\leq \sharp \Ed_{HR}(X)$ or, equivalently,
\begin{equation}\label{compHR}
\E^{\gamma}(X,V):=-\sharp \Ed^{\gamma}(X,V)\geq \sharp \E_{HR}(X,V).
\end{equation}

Now, we briefly discuss some qualitative properties of the ground states and of almost minimizers of the energy $\E^\gamma$, for $\gamma$ varying in the range $[0,1]$.
The significant case $\gamma=\frac{\sqrt 3}{2}$ will be  further analyzed in more details in Section \ref{rad32sec}.
\medskip

{\it The case $\gamma=0$.} This case is a slight generalization of the Heitmann-Radin model, being the dependence on the orientations only fictitious. 
Indeed,  for every given $N\in\N$ and $v\in \mathcal S^1$, let $\bar V:= \{v\}^N$ be a constant (arbitrarily chosen) orientation field. Then,  for every $(X,V)\in\AC_N$, 
we have that $\Ed^{0}(X,\bar V)=\Ed_{HR}(X)$, or, equivalently, $
\E^{0}(X, \bar V)=\E_{HR}(X,V)$. 
In particular, for configurations with constant orientation,  the class of minimizers (resp. almost minimizers) of $\E^0$ coincides with the class of minimizers (resp. almost minimizers) of $\E_{HR}$, analyzed in \cite{HeRa,DLF1,DLF2, AFS, DNP, FKS}.

\medskip
{\it The case $0<\gamma\le \frac 1 2$.} In such a range, for suitably chosen constant orientations, the ground states are the same as in the Heitmann-Radin model; in particular, they are subsets of the unitary triangular lattice.
Indeed, let $N\in\N$ and let $(X,V)$ be a minimizer of $\E_{HR}$ in $\AC_N$; by \cite{HeRa}, up to a rigid motion, $\X\subset\T^1$. Therefore, setting $\bar V:=\{(1;0)\}^N$, we have that $\Ed^{\gamma}(X,\bar V)=\Ed_{HR}(X)$, which, in view of \eqref{compHR}, yields the claim.
Notice also that not all constant orientations do the job:
Indeed, taking $(X,V)$ as above and $(X,\hat V)$ with  $\hat V:=\{(0;1)\}^N$, we have that all the ``horizontal bonds'' in $\Ed_{HR}(X)$, namely, the bonds $\{x_i,x_j\}$ with $x_j-x_i=\lambda (1;0)$ for some $\lambda\in  \{-1,1\}$, do not belong to $\Ed^\gamma(X,\bar V)$ and hence
$\E^\gamma(X,\hat V)>\E_{HR}(X,V)$ (see Figure \ref{VbarVhat}).

A natural question is whether, in the range $(0,\frac 1 2]$, 
for every configuration $(X,V)\in\AC$ there exists a configuration $( X,\bar V)$ (where $\bar V$ is a minimizer for fixed $X$)  such that
$\E^\gamma( X,\bar V)=\E_{HR}(X,V)$. 
%
%

\begin{figure}[h!]
{\def\svgwidth{300pt}
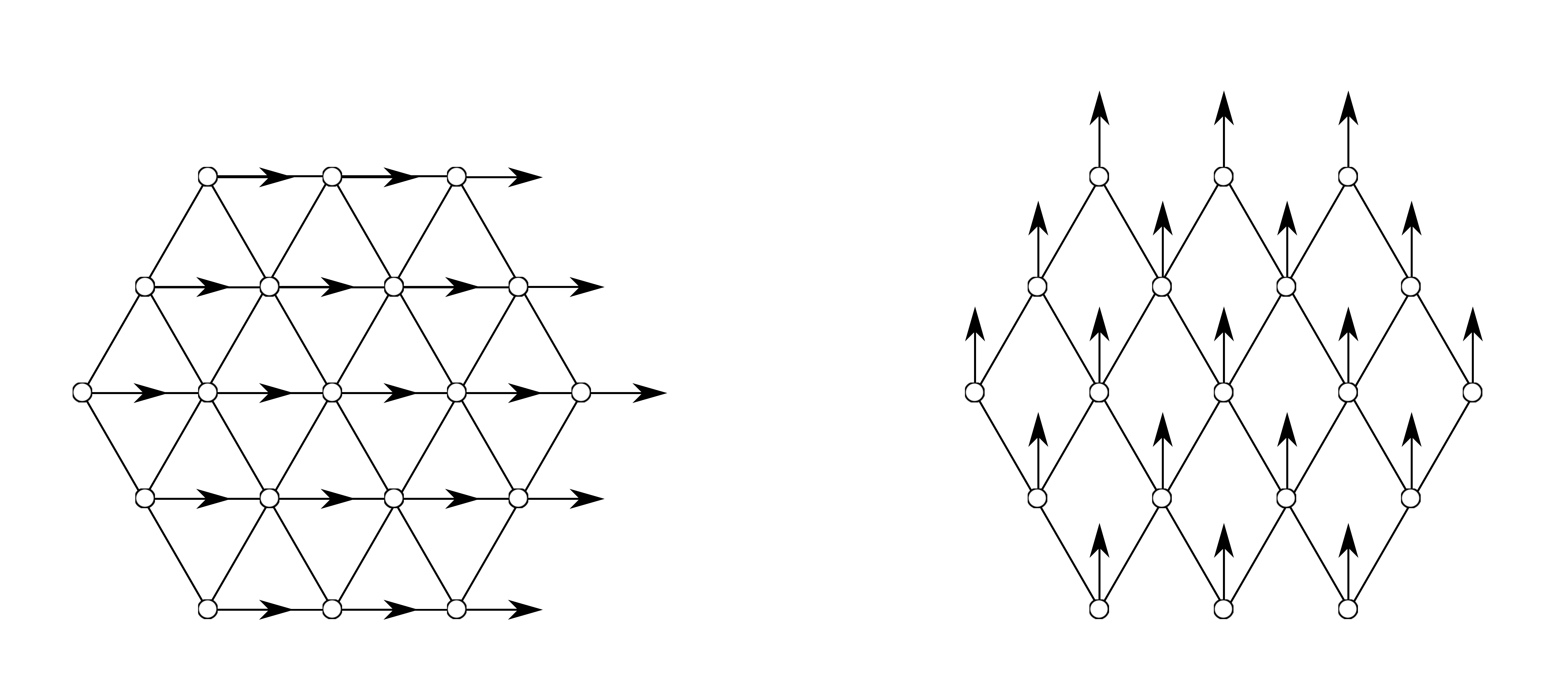}
\caption{The unique, up to rotation, minimizer $\X$ of the Heitmann-Radin energy for $N=19$. On the left: the configuration $(X,\bar V)$. On the right: the configuration $(X,\hat V)$.}	\label{VbarVhat}
\end{figure}


\medskip
{\it The case $\frac 1 2<\gamma\le \frac{\sqrt 3}{2}$.} In such a case we have that the maximal number of nearest neighbors of a point is equal to $4$. Indeed, let $(X,V)\in\AC$ and let
 $x\in \X$; without loss of generality we  can assume that $x=0$  and that the orientation of $x$ is given by $(1;0)$. Notice that if  $y\in\X$ is such that $\{x,y\} \in \Ed^\gamma(X,V)$, then $y=(\cos \theta; \sin \theta)$ for some $\theta\in (-\frac\pi 3,\frac \pi 3)\cup (\frac 2 3 \pi,\frac 4 3\pi)$. By the very definition of $\AC$, we immediately have that there can be at most two nearest neighbors of $x$ of the type $y=(\cos \theta; \sin \theta)$ with $\theta\in (-\frac\pi 3,\frac \pi 3)$ as well as  with $\theta\in (\frac 2 3 \pi,\frac 4 3 \pi)$; in particular, the claim follows.

Now, we show that the asymptotic (as $N\to +\infty$) energy per particle for minimizers equals to $-4$; more precisely, 
\begin{equation}\label{asymp}
-4N\le \min_{(X,V)\in\AC_N}\E^{\gamma}(X,V)\le -4N+C N^{\frac 1 2}\,,
\end{equation}
for some $C>0$ independent of $N$.
Notice that the first inequality in \eqref{asymp} is a direct consequence of the fact that $4$ is the maximal number of nearest neighbors; for what concerns the second inequality, for every $N\in\N$ let $V_N=\{(0;1)\}^N$ and let $\X_N$ be a subset of $\T^1$ with cardinality equal to $N$ such that the graph $\G^\gamma(X_N, V_N)$ is connected, the elements of $\Fa^\gamma(X_N,V_N)$ are all unitary rhombuses,  and
 $\sharp\partial^\gamma X_N\le CN^{\frac 1 2}$ for some $C>0$ independent of $N$. Such configurations can be easily constructed  (see, for instance, the configurations $(Y_N,W_N)$ provided by Definition \ref{Y_N} and Figure \ref{fig:YN} below) and satisfy
 \begin{equation}\label{am}
 \E^\gamma(X_N,V_N)\le -4\sharp(X_N\setminus\partial X_N)\le-4N+4C N^{\frac 1 2},
 \end{equation}
from which the second inequality in \eqref{asymp} follows.

Notice also that small perturbations of the configuration $(X_N,V_N)$ constructed above still yields almost minimizers for $\E^\gamma$ satisfying \eqref{am}. 
Indeed, 
for all $\rho>0$ set
$$
T^\rho:= 
\left(
\begin{array}{ll}
1+\rho & 0\\
0 & \frac{\sqrt{3 - \rho^2 - 2\rho}}{\sqrt{3}} 
\end{array}
\right) \, ,
$$
and notice that, for $\rho$ small enough, $T^\rho$ maps $\T^1$ into a unitary rhombic lattice.  
Now, for all $\rho>0$ let $X_N^\rho= (T^\rho(x_1), \ldots, T^\rho(x_N) )$, and notice that, for $\rho$ small enough, $(X^\rho_N,V_N)\in\AC_N$.   
Then, by an easy continuity argument, for all $\gamma\in (\frac 12, \frac{\sqrt{3}}{2})$ there exists $\rho>0$ (depending on $\gamma$) such that  $\E^\gamma(X^\rho_N,V_N)=\E^{\gamma}(X_N,V_N)$ and hence $(X^\rho_N,V_N)$ still satisfies \eqref{am}. Analogously, one can easily see that energy is invariant also under small perturbations of the orientation field $V_N$.
Finally,  for $N=4$ one can check that the configurations described above provide all the minimizers of $\E^\gamma$. 

\medskip
{\it The case $\frac{\sqrt 3}{2}<\gamma<1$.} We first show that in this case the maximal number of nearest neighbors of a point is equal to $2$.  Indeed, let $(X,V)\in\AC$ and let
 $x\in \X$; without loss of generality we  can assume that $x=0$  and that the orientation of $x$ is given by $(1;0)$. Notice that if  $y\in\X$ is such that $\{x,y\} \in \Ed^\gamma(X,V)$, then $y=(\cos \theta; \sin \theta)$ for some $\theta\in (-\frac\pi 6,\frac \pi 6)\cup (\frac 5 6 \pi,\frac 7 6\pi)$. By the very definition of $\AC$, we immediately have that there can be at most one nearest neighbor $y=(\cos \theta; \sin \theta)$ of $x$ with $\theta\in (-\frac\pi 6,\frac \pi 6)$ as well as  with $\theta\in (\frac 5 6 \pi,\frac 7 6 \pi)$; in particular, the claim follows.

For every $N\in\N$ we have
\begin{equation*}
-N\le \min_{(X,V)\in\AC_N}\E^\gamma(X,V)\le -N+1,
\end{equation*}
where the second inequality follows by considering the competitor $(\bar X,\bar V)\in\AC_N$ with $\bar X=((0;0), (1;0), \ldots, (N;0))$ and $\bar V=\{(1;0)\}^N$.

For such a range of the parameter $\gamma$, we can characterize the minimizers of the energy $\E^\gamma$ for every $N\in\N$.
To this purpose, we set $N_\gamma:=\lceil \frac{\pi}{\arccos\gamma}\rceil$, where $\lceil a\rceil:=\min\{m\in\N\cup\{0\}\,:\, m\ge a\}$ for every $a> 0$.

We first consider the case  $N<N_\gamma$. Let $(X,V)$ be a minimizer of $\E^\gamma$ in $\AC_N$. 
We claim that, up to a relabeling, $\X=\{x_1,\ldots,x_N\}$ where $\{x_j,x_k\}\in\Ed^{\gamma}(X,V)$ if and only if $|i-j|=1$ and hence $\E^\gamma(X,V)= -N+1$. 
To prove the claim, assume by contradiction that there exist $3\le J\le N$ and $y_1,\ldots,y_J\in\X$ such that $\{y_j,y_{j+1}\}\in\Ed^{\gamma}(X,V)$ for every $j=1,\ldots,J-1$,  $\{y_1,y_J\}\in\Ed^{\gamma}(X,V)$, and $C:=\bigcup_{j=1}^{J-1}[y_j,y_{j+1}]\cup [y_{J},y_{1}]$ is a simple and closed polygonal curve.
Then, denoting by $\alpha_j$ the inner angles of the polygon enclosed by $C$ and setting $\alpha_{\bar\jmath}:=\min_{j\in\{1,\ldots,J\}}\alpha_j$, by Euler formula, we have 
$$
J\alpha_{\bar\jmath}\le\sum_{j=1}^{J}\alpha_j=(J-2)\pi\,,
$$
so that
$$
\pi-2\arccos\gamma\le\alpha_{\bar\jmath}< \left(1-\frac{2}{\frac{\pi}{\arccos\gamma}}\right)\pi,
$$ 
which yields the desired contradiction.

Now we  consider the case  $N\ge N_\gamma$. Let $\X$ be the set of the vertices of a regular unitary $N$-gon centered at the origin, 
and for every $x\in \X$ let $v(x)$ be the unitary vector orthogonal to $x$, with orientation chosen so that either $v(x)\times x \equiv +1$ or $v(x)\times x \equiv -1$ for all $x\in\X$. Since $N\ge N_\gamma$, it easily follows that $(X,V) \in \AC_N$; moreover, $\E^\gamma(X,V)=-N$, so that $(X,V)$ is a minimizer. We easily conclude that if $(X,V)$ is any minimizer of $\E^\gamma$ in $\AC_N$, then every $x\in\X$ has exactly two nearest neighbors. 
It follows that every minimizer $(X,V)$ of $\E^\gamma$ in $\AC_N$ satisfies the following properties: there exist $K\in\N$, a partition $\{\X_1,\ldots,\X_K\}$ of $\X$ and  simple and closed curves $\Gamma_1,\ldots,\Gamma_K$ such that $\min_{k_1\neq k_2}\dist(\X_{k_1},\X_{k_2})\ge 1$,  $\X_k$ is the set of nodes of $\Gamma_k$ and
 $\sharp \X_k=\mathcal H^1(\Gamma_k)\ge N_\gamma$
  for every $k=1,\ldots,K$.
Moreover, for every $n=1,\ldots,N$ the angles $\hat x_n$ at the node $x_n$  satisfy  $|\hat x_n-\pi|<2\arccos\gamma$.
Viceversa, all the configurations above described, are admissible and have energy equal to $-N$, and hence provide all possible minimizers of $\E^\gamma$ in $\AC_N$.

\begin{figure}[h!]
{\def\svgwidth{300pt}
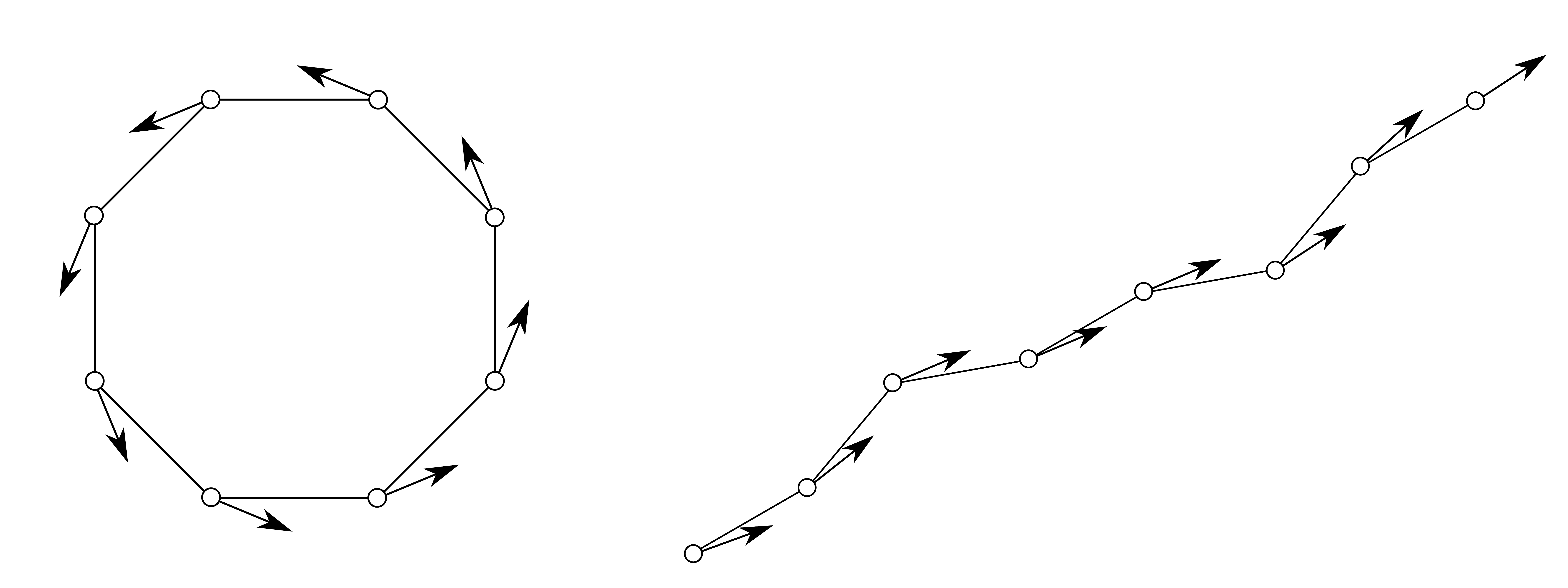}
\caption{On the left: A minimizer for $N\ge N_\gamma$. On the right: A minimizer for $N<N_\gamma$.}	\label{rad321}
\end{figure}

\medskip
{\it The case $\gamma=1$.} As for the range $(\frac{\sqrt 3}{2}, 1)$, we have that also in this case the maximal number of nearest neighbors of a point is equal to $2$. Moreover, given $(X,V)\in\AC$, we have that $\{x,y\}\in\Ed^1(X,V)$ if and only if $|x-y|=1$ and $v(x) = v(y)=\pm (x-y)$; it follows that the ground states of $\E^1$ in $\AC_N$ are made of $N$ aligned points forming a segment with constant tangent orientation, while the corresponding minimal energy is equal to $-N+1$. 


\begin{remark}
In looking at minimal configurations, one could first fix $\X$ and then minimize with respect to all possible orientation fields $V$, obtaining a reduced energy only depending on $\X$, and eventually further minimize with respect to $\X$. 
Notice also that further constraints could be enforced on $V$; for instance, denoting by $\bar V$ the average of the elements of $V$, one could plug some specific requirement on $|\bar V|$. Notice that $|\bar V|=1$ enforces constant orientation, while enforcing smallness conditions on $|\bar V|$ should favour cyclic-type configurations such as in the case $\gamma \in (\frac{\sqrt 3}{2},1)$ and large $N$ considered above.
Furthermore, in the energy functional a term penalizing variations of the orientation field could be added. One of the purposes of this paper is to show that, even in absence of such a penalization term, alignment of the orientation could be induced by the only minimization of $\E^\gamma$.       
\end{remark}


\section{Rhombic ground states for $\gamma=\frac{\sqrt 3}{2}$}\label{rad32sec}

This section is devoted to the detailed analysis of minimizers and almost minimizers of the energy $\E^\gamma$ defined in \eqref{energygamma} for $\gamma=\frac{\sqrt 3}{2}$. More precisely, we construct explicit rhombic minimizers (mimicking diamond formations) of the functional $\E^{\frac{\sqrt 3}{2}}$ in $\AC_N$ for every $N\in\N$  and 
we obtain a compactness result for almost minimizers of $\E^{\frac{\sqrt 3}{2}}$.

Since $\gamma$ is fixed to be equal to $\frac{\sqrt 3}{2}$, in this section we omit the dependence on $\gamma$ in all the notations introduced above, i.e., for every $(X,V)\in\AC$ we set $\Ed(X,V):=\Ed^{\frac{\sqrt 3}{2}}(X,V)$, $\E(X,V):=\E^{\frac{\sqrt 3}{2}}(X,V)$, $\G(X,V):=\G^{\frac{\sqrt 3}{2}}(X,V)$, $\partial X:=\partial^{\frac{\sqrt 3}{2}}X$, and so on.

We start by proving some geometric properties of the bond graphs of the configurations in $\AC$.

\subsection{Geometric properties of admissible configurations}\label{geoprop}
The following lemma provides geometric properties of the angles formed by two edges. Let $(X,V)\in\AC$.  For every $x_0,x_1,x_2\in \X$ with $\{x_0,x_1\},\,\{x_0,x_2\}\in\Ed(X,V)$\,, the symbol $\widehat{x_1x_0x_2}_-$  denotes the convex  angle formed by the segments $[x_0,x_1]$ and $[x_0,x_2]$ whereas $\widehat{x_1x_0x_2}_+$  denotes the concave one. Notice that $\widehat{x_1x_0x_2}_-\in(0,\pi]$\,, $\widehat{x_1x_0x_2}_+\in[\pi,2\pi)$ and
\begin{equation*}
\widehat{x_1x_0x_2}_-+\widehat{x_1x_0x_2}_+=2\pi\,.
\end{equation*}
\begin{lemma}\label{possibleanglesrmk}
Let $(X,V)\in\AC$\,. The following properties hold true.
\begin{itemize}
\item[(i)] If $x_0,x_1,x_2\in \X$ are such that $\{x_0,x_1\},\,\{x_0,x_2\}\in\Ed(X,V)$, then
\begin{equation}\label{possibleangles}
\begin{aligned}
\widehat{x_1x_0x_2}_-\in \Big\{\frac\pi 3\Big\}\cup \Big[\frac 2 3 \pi,\pi\Big] , \qquad\qquad\qquad
\widehat{x_1x_0x_2}_+\in\Big\{\frac 5 3 \pi\Big\}\cup \Big[\pi, \frac 4 3 \pi\Big].
\end{aligned}
\end{equation}
Moreover, if $\widehat{x_1x_0x_2}_-=\frac\pi 3$\,,
then $v(x_0)$ is parallel to the bisector of $\widehat{x_1x_0x_2}_{-}$ whereas, if $\widehat{x_1x_0x_2}_-=\frac 2 3\pi$\,, then $v(x_0)$ is orthogonal to the bisector of $\widehat{x_1x_0x_2}_{-}$\,.
\item[(ii)] If $x_0,x_1,x_2,x_3\in\X$ are such that $\{x_0,x_1\},\,\{x_0,x_2\},\,\{x_0,x_3\} \in\Ed(X,V)$, then 
$$
\widehat{x_1x_0x_2}_-\neq \widehat{x_2x_0x_3}_-.
$$
\item[(iii)] Every $x\in\X$ can lie on at most four bonds. Moreover, in such a case, the four segments corresponding to the four bonds containing $x$ 
form equal opposite angles $\frac\pi 3$ and $\frac 2 3\pi$\,.
\end{itemize}
\end{lemma}
\begin{proof}
We start by proving \eqref{possibleangles}.
Up to a rigid motion, we can assume that $x_0=(0;0)=0$ and that  $v(x_0)=(1;0)=1$.
By the very definition of $\Ed(X,V)$, we have that $x_1=(\cos\vartheta_1;\sin\vartheta_1)=e^{i\vartheta_1}$ and  $x_2=(\cos\vartheta_2;\sin\vartheta_2)=e^{i\vartheta_2}$ where $\vartheta_j\in [-\frac\pi 6, \frac\pi 6]\cup [\frac 5 6\pi,\frac 7 6 \pi]$ and by the very definition of $\AC$ we get $|\vartheta_1-\vartheta_2|\ge \frac\pi 3$. 
By combining the above conditions, \eqref{possibleangles} follows.

Now, if $|\vartheta_1-\vartheta_2|= \frac\pi 3$, then $(\vartheta_1;\vartheta_2)=\pm(-\frac\pi 6;\frac\pi 6)$,  which immediately implies that $v(x_0)$ is parallel to the bisector of the angle $\widehat{x_1x_0x_2}_-$.
Analogously, the case $|\vartheta_1-\vartheta_2|=\frac 2 3 \pi$ follows.

Now we prove (ii). Assume by contradiction that $
\widehat{x_1x_0x_2}_-= \widehat{x_2x_0x_3}_-=:\alpha$. By \eqref{possibleangles}, we have that either  $\alpha=\frac\pi 3$ or $
\alpha=\frac 2 3\pi$; since the  angles $\widehat{x_1x_0x_2}_-$, $\widehat{x_2x_0x_3}_-$ are adjacent, by the last statement in (i), we get incompatible conditions on $v(x_0)$.

Finally, we prove (iii). We use the same argument of Subsection \ref{examples}.
Without loss of generality we  can assume that $x=0$  and that $v(x)=(1;0)=1$. Notice that if  $y\in\X$ is such that $\{x,y\} \in \Ed(X,V)$, then $y=e^{i\vartheta}$ for some $\vartheta\in [-\frac\pi 6,\frac \pi 6]\cup [\frac 5 6 \pi,\frac 7 6\pi]$. 
By the very definition of $\AC$, we immediately have that there can be at most two nearest neighbors of $x$ of the type $y=e^{i\vartheta}$ with $\vartheta\in [-\frac\pi 6,\frac \pi 6]$ and if there are exactly two nearest neighbors in such an interval they should be $x_1=e^{i\frac\pi 6}$ and $x_2=e^{-i\frac\pi 6}$. 
Similarly,  there can be at most two nearest neighbors of $x$ for  $\theta\in [\frac 5 6\pi,\frac 7 6\pi ]$ and if there are exactly two nearest neighbors in such an interval they should be $x_3=e^{i\frac 5 6 \pi }$ and $x_4=e^{i\frac 7 6 \pi}$.

\end{proof}
\begin{lemma}\label{notri}
Let $(X,V)\in\AC$. Then $\Per(f)\ge 4$ for every  $f\in \Fa(X,V)$\,.
\end{lemma}
\begin{proof}
Assume by contradiction that there exists a face $\bar f\in\Fa(X,V)$ having perimeter equal to 3. By definition, $\bar f$ is an equilateral triangle with unitary side-length; in particular, denoting by $x_1,\,x_2,\,x_3$ the vertices of $\bar f$ and by $\hat x_j$ ($j=1,2,3$) the inner angles of $\bar f$\,, we have that $\hat x_j=\frac\pi 3$ for every $j=1,2,3$\,.

By the very definition of $\Ed(X,V)$ we have that $\langle v(x_j),v(x_{k})\rangle \ge 0$ for every $j\neq k$.
Since, by Lemma \ref{possibleanglesrmk}(i), $v(x_j)$ is parallel to the bisector of $\hat x_j$ for every $j=1,2,3$, we get a contradiction.
 

\end{proof}
\begin{lemma}\label{onlyrhombic}
Let $(X,V)\in\AC$ and let $f\in \Fa(X,V)$. If $\Per(f)=4$ then $f$ is a rhombus with side-length equal to one and inner angles equal to $\frac{\pi}{3},\frac{2}{3}\pi$.
In such a case, the orientations of the vertices of $f$ are all equal to each other. Moreover, such an orientation is orthogonal to the shortest diagonal of the rhombus and hence parallel to the longest one. 
\end{lemma}
\begin{proof}

Since $f$ is a rhombus, the sum of two consecutive angles equals to $\pi$\,; hence, by Lemma \ref{possibleanglesrmk}(i), we have that the inner angles of $f$ should be equal to $\frac{\pi} {3},\frac{2}{3}\pi$\,.

Moreover, by Lemma \ref{possibleanglesrmk}(ii) we deduce that all the orientations should be parallel to the bisectors of the $\frac\pi 3$ angles\,, that in turns are parallel to the longest diagonal of $f$\,.
Finally, since by the very definition of $\Ed(X,V)$, the scalar product between any pair of such orientations should be non-negative, we get that they are all equal to each other.
\end{proof}

In view of Lemma \ref{onlyrhombic}, it is reasonable to denote by $\Fsq(X,V)$ the set of faces $f$ with $\Per(f)=4$\,. Note that in this case also $\Per_{\gr}(f)=4$.
 Moreover, again by Lemma \ref{onlyrhombic}, we can define the orientation $v(f)$ of a face $f\in\Fsq(X,V)$ as the orientation of its vertices.

\subsection{Geometric decomposition of the energy}

Here we decompose the energy $\E$ in \eqref{defE0} into the sum of a volume term (which actually depends only on the number of points of the configuration) and a perimeter type term.

\begin{proposition}\label{decompprop}
For every $(X,V)\in\AC$, it holds
\begin{equation}\label{decompformula}
\begin{aligned}
\E(X,V)=&-2\sharp \X+\frac 1 2\defe_\gr(X,V)+\frac 1 2\Per_\gr(X,V)+2\chi(X,V)\,,
\end{aligned}
\end{equation}
where we have set $\defe_\gr(X,V)\defl\sum_{f\in \Fa^\bad(X,V)}(\Per_\gr(f)-4)$ and
\begin{equation*}
\Fbad(X,V)\defl \Fa(X,V)\setminus\Fsq(X,V).
\end{equation*}
\end{proposition}
\begin{proof}
The proof relies on the strategy of \cite[Theorem 2.1]{DLF1}, with relevant differences due to the rhombic lattice we are dealing with.
By triangulating each of the faces in $\Fbad(X,V)$ as in Lemma \ref{triang}, one can construct a new planar graph $\overline{\G}(X,V)=(\X,\overline{\Ed}(X,V))$, all of whose faces are either triangles or rhombuses (the latter with unitary side-length). We denote by $\overline{\Fa}(X,V)$ the set of the faces of $\overline{\G}(X,V)$. Therefore,
\begin{equation}\label{split_faces}
\overline{\Fa}(X,V)=\Fsq(X,V)\cup\overline{\Fa}^\triangle(X,V),
\end{equation}
where $\overline{\Fa}^\triangle(X,V)$ denotes the set of triangular faces of $\overline{\G}(X,V)$.
Note that, in view of Lemma \ref{notri}, $\overline{\Fa}^\triangle(X,V)$ accounts only for the triangles obtained by triangulating the faces in $\Fbad(X,V)$. Moreover, we denote by $\Add(X,V):=\overline{\Ed}(X,V)\setminus \Ed(X,V)$ 
the set of the additional edges due to triangulation and by $\overline{\Ed}^\inte(X,V)$ the set of the interior edges of $\overline{\G}(X,V)$;  
thus, in view of Lemma \ref{triang}, we have
%
\begin{equation}\label{zero}
\begin{aligned}
\overline{\Ed}^{\mathrm{int}}(X,V)&=\Ed^{\mathrm{int}}(X,V)\cup\Ed^{\Wi,\mathrm{int}}(X,V)\cup \Add(X,V);\\
\sharp\Add(X,V)&=\sum_{f\in \Fbad(X,V)}(\Per_{\gr}(f)-3);\\
\sharp \overline{\Fa}^\triangle(X,V)&=\sum_{f\in\Fbad(X,V)}(\Per_{\gr}(f)-2).
\end{aligned}
\end{equation}
By construction, the set of the boundary edges and the set of the exterior wire edges do not change under the triangulation procedure, and hence the  graph perimeter of $\G(X,V)$ coincide with the graph perimeter $\overline{\G}(X,V)$.
Moreover, the same holds for the Euler characteristic of ${\G}(X,V)$  in the sense of definition \eqref{Euchar} , i.e., $\chi(\overline{\G}(X,V))=\chi(X,V)$.
Notice also that, by the classification of the edges in Section \ref{prelimgra}, we have 
\begin{equation}\label{classed}
\sharp\Ed^{\mathrm{int}}(X,V)+\sharp \Ed^{\Wi,\mathrm{int}}(X,V)= \sharp \Ed(X,V)-\sharp \Ed^{\Wi,\mathrm{ext}}(X,V)-\sharp \Ed^\partial (X,V).
\end{equation}

%
We note that every  exterior wire edge in $\overline{\G}(X,V)$ participates to no face, every boundary edge participates to one face, and all other edges participate to two faces. 
As a consequence, using \eqref{split_faces}, \eqref{zero} and \eqref{classed}, we thus find
\begin{equation}\label{second}
\begin{aligned}
&\phantom{=}4\sharp \Fsq(X,V)+3\sharp\overline{\Fa}^{\triangle}(X,V)\\
&=\sum_{f\in  \Fsq(X,V)}\Per_{\gr}(f)+\sum_{f\in  \overline{\Fa}^{\triangle}(X,V)}\Per_{\gr}(f)=2\sharp \overline{\Ed}^{\mathrm{int}}(X,V)+\sharp\Ed^{\partial}(X,V)\\
&=2\sharp\Ed^{\mathrm{int}}(X,V)+2\sharp \Ed^{\Wi,\mathrm{int}}(X,V)+2\sharp  \Add(X,V)+\sharp\Ed^{\partial}(X,V)\\
&=2\sharp\Ed(X,V)-2\sharp \Ed^{\Wi,\ext}(X,V)-\sharp\Ed^{\partial}(X,V)+2\sharp  \Add(X,V).
\end{aligned}
\end{equation}
Moreover, by \eqref{zero}, we get
\begin{equation}\label{second2}
\begin{aligned}
2\sharp  \Add(X,V)-\sharp \overline{\Fa}^{\triangle}(X,V)&=\sum_{f\in\Fbad(X,V)}(\Per_{\gr}(f)-4)
=\defe_{\gr}(X,V).
\end{aligned}
\end{equation}
By \eqref{second} and \eqref{second2}, we obtain
\begin{equation}\label{second3}
\begin{aligned}
4\sharp\overline{\Fa}(X,V)=&4\sharp \Fsq(X,V)+3\sharp\overline{\Fa}^{\triangle}(X,V)+\sharp\overline{\Fa}^{\triangle}(X,V)\\
=&2\sharp\Ed(X,V)-2\sharp \Ed^{\Wi,\ext}(X,V)-\sharp\Ed^{\partial}(X,V)+2\sharp  \Add(X,V)\\
&+2\sharp  \Add(X,V)-\defe_{\gr}(X,V)\\
=&2\sharp\Ed(X,V)-\Per_{\gr}(X,V)-\defe_\gr(X,V)+4\sharp  \Add(X,V),\\
\end{aligned}
\end{equation}
where the last equality follows by \eqref{perimeter}.
By \eqref{Euchar} and \eqref{second3} we deduce that
\begin{equation*}
\begin{aligned}
4\chi(X,V)=& 4\sharp\X-4\sharp\overline{\Ed}(X,V)+4\sharp \overline{\Fa}(X,V)\\
=&4\sharp \X-4\sharp\Ed(X,V)-4\sharp\Add(X,V)\\
&+2\sharp\Ed(X,V)-\Per_\gr(X,V)-\defe_\gr(X,V)+4\sharp  \Add(X,V)
\\
=&4\sharp \X-2\sharp\Ed(X,V)-\Per_\gr(X,V)-\defe_\gr(X,V),
\end{aligned}
\end{equation*}
which yields
\begin{equation*}
-\sharp\Ed(X,V)=-2\sharp \X+\frac 1 2\defe_\gr(X,V)+\frac 1 2\Per_\gr(X,V)+2\chi(X,V),
\end{equation*}
whence the claim follows since, by definition, $\E(X,V)=-\sharp\Ed(X,V)$.
\end{proof}
From now on, for every $(X,V)\in\AC$, we set
\begin{equation}\label{surface+}
\F(X,V):=\frac 1 2 \defe_{\gr}(X,V)+\frac 1 2\Per_{\gr}(X,V)+2\chi(X,V),
\end{equation}
so that by \eqref{decompformula}, we get $\E(X,V)=-2\sharp\X+\F(X,V)$.


\subsection{Diamond formations}
Here we provide rhombic minimizers of the energy $\E$ in $\AC_N$ for every $N\in\N$; more precisely, for every $N\in\N$ we construct a minimizing configuration $(Y_N,W_N)$ in $\AC_N$, referred to as {\it canonical configuration}, such that $\Y_N$ is a subset of the triangular lattice $\T^1$ defined in \eqref{T1} and $\Fbad(Y_N,W_N)=\emptyset$.




In order to define $(Y_N,W_N)$, we preliminarily note that for every $N\in\N$ there exists a unique pair $(l;\eta)\in(\N\cup\{0\})\times(\N\cup\{0\})$ with $0\le\eta<2l+3$ such that
\begin{equation*}
N=(l+1)^2+\eta\,.
\end{equation*}

For every $\ell\in\N$, let $R_\ell$ be the closed rhombus with inner angles equal to $\frac{\pi}{3}\,,\frac{2}{3}\pi$, and whose longest diagonal is $\{0\}\times [0,\sqrt{3}\ell]$. Given $l\in\N$, let $\Gamma_l:= \partial R_{l+1}  \cap \{ (x,y) \in\R^2: \, y\ge \frac{\sqrt 3 }{2} (l+1)\}$. 
Notice that $\Gamma_l$ has length equal to $2l +2$, and hence it is the support of a  curve $\gamma_l$ parametrized on $[0,2l+2]$ and with unitary velocity. 

\begin{definition}\label{Y_N}
We set $\Y_0:=\emptyset$, $\Y_1:=\{(0;0)\}$\,, $\Y_2:=\Y_1\cup\{(\frac 1 2;\frac{\sqrt 3}{2})\}$\,, $\Y_{3}:=\Y_2\cup\{(0;\sqrt 3)\}$\,.
If $N=(l+1)^2$ for some $l\in\N$ we define the $\Y_N:=R_{l}\cap\T^1$, while if $N=(l+1)^2+\eta$, for some $l,\eta\in\N$ with $1\le \eta \le 2l+2$ we set $\Y_N:= (R_{l}\cap\T^1) \cup \{\gamma_l(0)\} \ldots \cup\{\gamma_l(\eta -1)\}$\,. 
We set $W_N:=\{(0;1)\}^N$ and
we call $(Y_N,W_N)$ the canonical configuration in $\AC_N$.
\end{definition}
%
\begin{figure}[h!]
{\def\svgwidth{300pt}
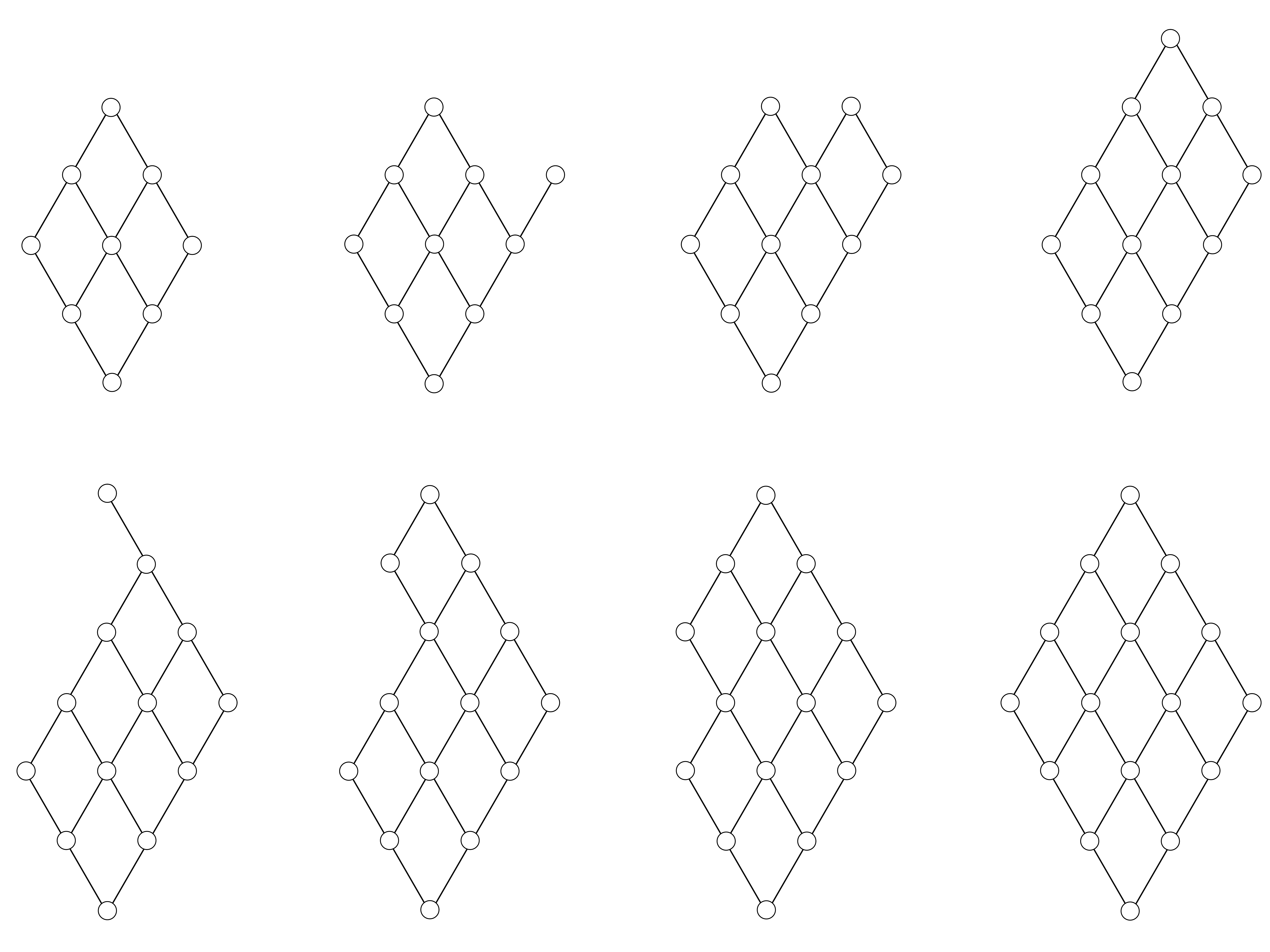}
\caption{The graphs $\G(Y_N,W_N)$ of the configurations $(Y_N,W_N)$ for $N=9,\ldots, 16$.}	\label{fig:YN}
\end{figure}

Note that $\Per_\gr(Y_0,W_0)=0$ and for every $N\in\N$
\begin{equation}\label{percan}
\Per_{\gr}(Y_N,W_N)=\begin{cases}
4l&\textrm{ if }N=(l+1)^2\,,\\
4l+2&\textrm{ if }N=(l+1)^2+\eta\textrm{ with }1\le\eta\le l+1\,,\\
4l+4&\textrm{ if }N=(l+1)^2+\eta\textrm{ with }l+2\le\eta\le 2l+2\,.
\end{cases}
\end{equation}
Moreover, by construction, $\defe_\gr(Y_N,W_N)=0$ and $\G(Y_N,W_N)$ is connected, so that
\begin{equation}\label{forserichiam}
\F(Y_N,W_N)=\frac 1 2 \Per_{\gr}(Y_N,W_N)+2.
\end{equation}

%
%


Now, we fix some notations that will be useful to prove that the configurations $(Y_N,W_N)$ are minimizers. We recall that $\partial X$ is defined according to Definition \ref{borX0} for $\gamma=\frac{\sqrt 3}{2}$.
\begin{definition}\label{borX}
For every $N\in\N$ and for every $(X,V)\in\AC_N$\,, we set $N':=N-\sharp\partial X$\,, $\X':=\X\setminus\partial X$ and we denote by $(X',V')$ the configuration in $\AC_{N'}$ satisfying the following property: for every $j=1,\ldots,N'$ there exists a unique $k\in\{1,\ldots,N\}$ such that
$x'_j=x_k$, $v'_j=v_k$, $x_k\notin\partial X$\,.
\end{definition}

Notice that, for every $N,\, N', \, P\in \N\cup\{0\}$  
with $N> N'$ there exists a unique pair $(k;\delta)\in(\N\cup\{0\})\times (\N\cup\{0\})$ with 
$\delta<P+8(k+1)$ such that 
\begin{equation}\label{Ndecom}
\begin{aligned}
N=&N'+(P+8)+\ldots+(P+8k)+\delta\\
=&N'+k(P+4(k+1))+\delta\,.
\end{aligned}
\end{equation}

We are now in a position to state and prove the main result of this section; the proof relies on some auxiliary lemmas provided in Appendix \ref{aux2}.
\begin{theorem}\label{fac}
For every $N\in\N$, it holds
\begin{equation}\label{newii}
\E(Y_N,W_N)=\min_{(X,V)\in\AC_N}\E(X,V)\,.
\end{equation}
\end{theorem}
\begin{proof}
The result is true for $N\in\{1,\ldots, 7\}$ by Lemma \ref{inspection}.
Let $\bar N\ge 7$\,. We will proceed by induction on $N$ and assume that $(Y_N,W_N)$ is a minimizer of $\F$ in $\AC_{N}$ for every $N=1,\ldots,\bar N$\,. Let now $N=\bar N+1(\ge 8)$ and let $(X,V)$ be a minimizer of $\F$ in $\AC_N$\,.  
In view of \eqref{decompformula} we have that $\E(X,V)=-2\sharp \X+\F(X,V)$\,, where $\F$ is defined in \eqref{surface+}.
It follows that $(X,V)$ 
minimizes $\F$ in $\AC_N$.

By Lemma \ref{lemma0} and by the inductive assumption, we have that $\G(X,V)$ is connected so that by Lemma \ref{lm:bouchar} we have
\begin{equation}\label{banal}
\begin{aligned}
\F(X,V)=&\frac 1 2 \Per_{\gr}(X,V)+\frac 1 2 \defe_\gr(X,V)+2\\
= &\frac 1 2 \sharp\partial X+\frac 1 2 (\Per_{\gr}(X,V)-\sharp\partial X)+\frac 1 2 \defe_\gr(X,V)+2
\ge \frac 1 2 \sharp\partial X+2 \,,
\end{aligned}
\end{equation}
where the last inequality follows by Lemma \ref{puntidibordo} and by the fact that $\defe_\gr(X,V)\ge 0$.
Let $(X',V')$ be defined as in Definition \ref{borX}. Then, $(X',V')\in\AC_{N'}$ with $N'=N-\sharp\partial X<N$\,. 
Now, in view of \eqref{Ndecom} applied with $P=\Per_{\gr}(Y_{N'},W_{N'})$, there exists 
a unique pair $(k;\delta)\in(\N\cup\{0\})\times (\N\cup\{0\})$ with 
$\delta<\Per_{\gr}(Y_{N'},W_{N'})+8(k+1)$ such that
\begin{equation}\label{deco0}
\sharp\partial X=N-N'=k(\Per_{\gr}(Y_{N'},W_{N'})+4(k+1))+\delta\,.
\end{equation}
We consider three different cases.
\vskip5pt
{\it Case 1: $k\ge 2$\,.}
By \eqref{banal} and \eqref{deco0} we have
\begin{equation}\label{ca1}
\F(X,V)\ge \frac 1 2\sharp\partial X+2= \frac 1 2\Big(k(\Per_{\gr}(Y_{N'},W_{N'})+4(k+1))+\delta\Big)+2\,.
\end{equation}
By Lemma \ref{shellconstr}, since $k\ge 2$, we have
$$
 \Per_{\gr}(Y_N,W_N)\le \Per_{\gr}(Y_{N'},W_{N'})+8(k+1),
$$
so that, using again that $k\ge 2$\,, we obtain
\begin{equation}\label{ca100}
\begin{aligned}
k(\Per_{\gr}(Y_{N'},W_{N'})+4(k+1))+\delta \ge&
\Per_{\gr}(Y_{N'},W_{N'})+8(k+1)
\ge \Per_{\gr}(Y_N,W_N)\,;
\end{aligned}
\end{equation}
in view of \eqref{ca1}, \eqref{ca100}, and  \eqref{forserichiam}, we deduce that $
\F(X,V)\ge \F(Y_N,W_N)$ which implies \eqref{newii}. 

\vskip5pt
{\it Case 2: $k=1$.}
By \eqref{banal}, \eqref{deco0}, Lemma \ref{shellconstr} and \eqref{forserichiam}, we have
\begin{equation}\label{ca2}
\begin{aligned}
\F(X,V)\ge &\frac 1 2\sharp\partial X+2= \frac 1 2\Big(\Per_{\gr}(Y_{N'},W_{N'})+8+\delta\Big)+2\\
\ge &  \frac 1 2\Big(\Per_{\gr}(Y_{N'},W_{N'})+8+r_{N',N}\Big)+2-\frac 1 2\\
=&  \frac 1 2 \Per_{\gr}(Y_{N},W_{N})+2-\frac 1 2=\F(Y_N,W_N)-\frac 1 2\,.
\end{aligned}
\end{equation}
Since $\F$ takes values in $\N$\,, by \eqref{ca2} we deduce \eqref{newii} also in this case. 

\vskip5pt
{\it Case 3: $k=0$\,.}
By \eqref{forserichiam} and by Lemma \ref{shellconstr} we have
\begin{equation*}\label{basic0}
\F(Y_N,W_N)=\frac 1 2  \Per_{\gr}(Y_{N},W_{N})+2\le \frac 1 2 \Big(\Per_{\gr}(Y_{N'},W_{N'})+8\Big)+2=\F(Y_{N'},W_{N'})+4\,;
\end{equation*}
therefore, by Lemma \ref{increase} and by the inductive assumption, we obtain
\begin{equation*}
\F(X,V)\ge \F(X',V')+4\ge \F(Y_{N'},W_{N'})+4\ge\F(Y_N,W_N)\,,
\end{equation*}
which by minimality of $(X,V)$ implies \eqref{newii} also in this case.
\end{proof}
\begin{remark}
A general rhombic crystallization result for all the minimizers does not hold.
Indeed, for all $N$ of the form $N=(l+1)^2+1$ or $N=l(l+1)+1$ with $l\ge 1$, one can construct a minimizer having a face $\bar f$ with $\Per_{\gr}(\bar f)=5$, as shown in Figure \ref{fig:penta} below. Notice also that for such minimizers the orientation is not constant.
\end{remark}
\begin{figure}[hh]
{\def\svgwidth{250pt}
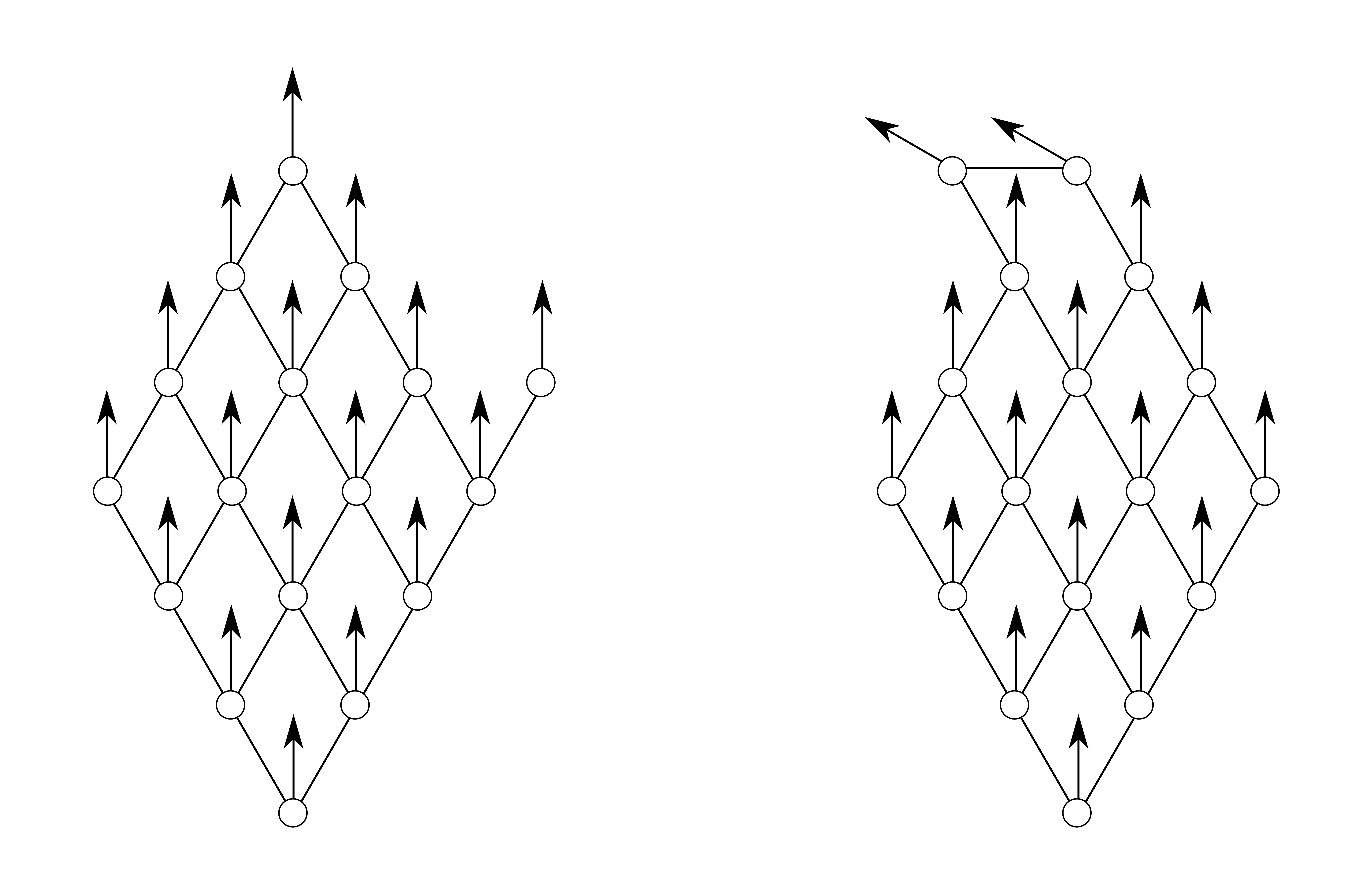}
\caption{Two minimizers of $\E$ in $\AC_{17}$. On the left: the canonical minimizer. On the right: A minimizer having a non-rhombic face.}	\label{fig:penta}
\end{figure}

\subsection{Compactness of quasi-minimizers}

This section is devoted to a compactness result for quasi-minimizers of $\E=\E^{\frac{\sqrt 3}{2}}$.
We start by considering the class of empirical measures associated to the configurations in $\AC$.  
To this purpose, let $\mathcal{M}$ denote the class of vectorial Radon measures from $\R^2$ to $\R^2$; for every $N\in\N$ we set 
\begin{equation*}
\begin{aligned}
\mathcal{EM}_N&:=\left\{ \mu\in\mathcal{M}\ :\ \mu=\sum_{i=1}^Nv_i\,\delta_{\frac{x_i}{\sqrt N}} 
\textrm{ for some }\ (X,V)\in\AC_N   \right\}.\\
\end{aligned}
\end{equation*}
Fix $N\in\N$.
Since there is a natural one-to-one correspondence $\mathcal{I}_N$ from $\mathcal{EM}_N$ to the class $\AC_N$, we can re-define the energy \eqref{defE0} on $\mathcal{M}$ introducing the functional $\E_{N}:\mathcal{M}\to(-\infty,+\infty]$ given by
\begin{equation*}
\E_{N}(\mu)\defl\begin{cases} \E(\mathcal{I}_N(\mu)) & \textnormal{if}\ \mu\in\mathcal{EM}_N\\
+\infty & \textnormal{otherwise}.
\end{cases}
\end{equation*}
In the following, with a little abuse of notation, for every $\mu\in\mathcal{EM}_N$ we will write all the objects introduced in Subsection \ref{bonddef} in terms of $\mu$, namely, we set $\Ed(\mu):=\Ed(\I_N(\mu))$, $\Fa(\mu):=\Fa(\I_N(\mu))$ and so forth.  In all this subsection, the symbol $\Per^{\dg}$ will denote the De Giorgi's perimeter; recalling the notation introduced in \eqref{perimeter}, \eqref{graphperG} and \eqref{graphgeoper}, we have 
\begin{equation}\label{DeG}
\begin{aligned}
\Per(\mu)=&\Per(\I_N(\mu))=\Per^{\dg}(O(\mu)),\\
\Per_\gr(\mu)=& \Per_\gr(\I_N(\mu))=\Per^{\dg}(O(\mu))+2\sharp \Ed^{\Wi,\ext}(\mu),\\
\Per_\gr(f)=&\Per^{\dg}(f)+2\sharp\Ed^{\Wi,\inte}(f)\quad\textrm{ for every face }f\in\Fa(\mu).
\end{aligned}
\end{equation}

Moreover,  setting 
\begin{equation*}
f_N:=\frac{f}{\sqrt N}\qquad\textrm{for every } f\in\Fa(\mu), 
\end{equation*}
we have
\begin{equation*}
\Per^\dg(f_N)=\frac{\Per^\dg(f)}{\sqrt{N}}\qquad\textrm{
for every }f\in\Fa(\mu),
\end{equation*} 
so that
\begin{equation}\label{lbbad}
\Per^\dg(f_N)\ge \frac{5}{\sqrt{N}}\qquad\textrm{
for every }f\in\Fbad(\mu);
\end{equation}
furthermore, setting $O_N(\mu):=\frac{1}{\sqrt{N}}O(\mu)=\bigcup_{f\in\Fa(\mu)} f_N$, it holds
\begin{equation*}
\Per^\dg(O_N(\mu))=\frac{\Per^\dg( O(\mu))}{\sqrt N}.
\end{equation*}

Therefore, by
\eqref{decompformula} and \eqref{DeG}, the functionals $\E_N$ can be decomposed
\begin{equation}\label{epdecompformula}
\begin{aligned}
\E_N(\mu)+2\mu(\R^2)=&\frac{\sqrt N}{2} \sum_{f\in \Fa^\bad(\mu)}\left(\Per^\dg(f_N)-\frac{4}{\sqrt N}\right)\\
&+\frac{\sqrt N}{2}\Per^\dg(O_N(\mu))+\sharp \Ed^{\Wi}(\mu)+2\chi(\mu)\,.
\end{aligned}
\end{equation}

We can finally state the desired compactness result.
\begin{theorem}\label{compthm}
Let $\{\mu_N\}\subset\M$ be such that 
\begin{equation}\label{enbou}
\E_N(\mu_N)+{2}\,\mu_N(\R^2)\le C \sqrt{N} \textrm{ for some } C\in\R\,.
\end{equation}
Then,  up to a subsequence, $\frac{\sqrt 3}{2N}\mu_N\weakstar\mu$ for some $\mu\in\mathcal{M}$. Moreover, $\mu=\sum_{j\in J}v_j\Indi_{\omega_j}\ud x$ where $J\subseteq\N$\,, $v_j\in\mathcal{S}^1$ for every $j\in J$, and $\sum_{j\in J}\Per^{\dg}(\omega_j)<+\infty$.
\end{theorem}
Before providing the proof of Theorem \ref{compthm} we state the following compactness property  which is a corollary of Ambrosio compactness result \cite{A} and of \cite[Theorem 4.23]{AFP}. 
\begin{proposition}\label{SBVcompactProp}
Let $\{u_h\}\subset SBV(\R^2;\R^2)$. If there exists $p>1$ and $C>0$ such that
\begin{equation}\label{ipot}
\int_{\R^2} |\nabla u_h |^p \ud x+ \mathcal{H}^1(S_{u_h})+ \left\lVert u_h\right\rVert_{L^{\infty}(\R^2)}\leq C\ \text{for all}\ h\in\N\,,
\end{equation}
then, up to a subsequence, there exists $u\in SBV(\R^2)$ such that 
for every open bounded set $A\subset\R^2$ it holds 
\begin{equation*}
\begin{split}
 & u_h\to u\ \text{in}\ L^1(A;\R^2)\,,\\
 & \nabla u_h \rightharpoonup \nabla u\ \text{in}\ L^1(A;\R^{2\times 2})\,,\\
& \liminf_{h\to\infty} \mathcal{H}^1(S_{u_h}\cap A)\geq \mathcal{H}^1(S_{u}\cap A).
\end{split}
\end{equation*}
Moreover, if $\nabla u= 0$, then $u=\sum_{j\in J}v_j\Indi_{\omega_j}\ud x$ where $J\subseteq\N$\,, $v_j\in \R^2$ for every $j\in J$, and $\sum_{j\in J}\Per^{\dg}(\omega_j)<+\infty$.
\end{proposition}

In the following we say that a sequence $\{u_h\}_h$ converges to some function $u$ in $SBV_{\loc}(\R^2;\R^2)$ if $u_h\to u$ in $L^1_{\loc}(\R^2;\R^2)$ and  $\{u_h\}_h$ satisfies \eqref{ipot} for some $p>1$.

\begin{proof}[Proof of Theorem \ref{compthm}]
The proof is divided into two steps.

{\it Step 1.}
Let $\{\rho_N\}_{N\in\N}$ be the sequence defined by 
$$
\rho_N:=\sum_{f\in \Fsq(\mu_N)}v(f)\Indi_{f_N}\qquad\textrm{ for every }N\in\N\,;
$$
we first claim that, up to a subsequence, $\rho_N\weakly \rho$ in $SBV_{\loc}(\R^2;\R^2)$ for some function $ \rho$ of the form $\rho:=\sum_{j\in J}v_j\Indi_{\omega_j}$ where $v_j$ and $\omega_j$ satisfy the claim of the theorem.

In view of \eqref{enbou}, \eqref{epdecompformula}, \eqref{lbbad}, and \eqref{bouchar}, we have 
\begin{equation}\label{perbound}
\begin{aligned}
C & \ge \frac{1}{\sqrt N} ( \E_N(\mu_N)+{2}\,\mu_N(\R^2)) \ge \frac 1 2
\Per^\dg(O_N(\mu_N))+ {\frac{1}{10}} \sum_{f\in \Fbad(\mu_N)}  \Per^\dg(f_N) 
\\
& \ge \frac{1}{10}  \Per^\dg \Big(\bigcup_{f\in \Fsq(\mu_N)}f_N\Big) = \frac{1}{10} \Per^\dg(\Om_N(\mu_N)),
\end{aligned}
\end{equation}
where we have denoted by 
$\Om_N(\mu_N)$ the interior of 
$\bigcup_{f\in \Fsq(\mu_N)}\clos(f_N)$\,.
By the very definition of $\rho_N$ we have that $\|\rho_N\|_{L^\infty}\le 1$,  $\nabla\rho_N=0$, and $\mathcal H^1(S_{\rho_N})\le \Per^\dg(\Om_N(\mu_N))\le 10 C$. 
 

%
Then, the claim directly follows by Proposition \ref{SBVcompactProp}\,.

{\it Step 2.}
Now we prove that, up to a subsequence,
\begin{equation}\label{fina}
\frac{\sqrt 3}{2N}\mu_N\weakstar \rho\ud x\qquad\textrm{ as }N\to +\infty,
\end{equation}
with $\rho$ provided by Step 1.
To this purpose, for every $f\in \Fsq(\mu_N)$ we denote by $a_j(f_N)$ ($j=1,2,3,4$) the vertices of $f_N$\,,  and we define
\begin{equation}
\hat\mu_N:={\frac 1 {4N}}\sum_{f\in \Fsq(\mu_N)}v(f)\sum_{j=1}^4 \delta_{{a_j(f_N)}}\,,
\qquad\tilde\mu_N:=\rho_N\ud x\,.
\end{equation}

By \eqref{perbound} and by the isoperimetric inequality, we obtain
\begin{equation*}
\int_{\R^2}|\rho_N|\ud x=|\tilde\mu_N|(\R^2)=|\Omega_N(\mu_N)|\le C\,.
\end{equation*}
We now show that $\frac{\sqrt{3}}{2}\hat \mu_N - \tilde \mu_N \weakstar 0$\,. Let $\psi\in C^0_c(\R^2)$\,; for every $f\in\Fsq(\mu_N)$, let $\psi_{f_N}$ be the average of $\psi$ on $f_N$\,. Then
\begin{equation}\label{cappmenotil}
\begin{aligned}
\left | \langle \frac{\sqrt{3}}{2} \hat \mu_N - \tilde \mu_N, \psi\rangle \right |&= 
\left | \sum_{f\in \Fsq(\mu_N)} \langle  \frac{\sqrt{3}}{2} \hat \mu_N - \tilde \mu_N, \psi \res \clos(f_N)\rangle\right |
\\
&= \left |  \sum_{f\in  \Fsq(\mu_N)} \langle   \frac{\sqrt{3}}{2} \hat \mu_N - \tilde \mu_N, (\psi - \psi_{f_N}) \res \clos(f_N)\rangle \right | \\
&\le 
 \sum_{f\in  \Fsq(\mu_N)} \left| \langle \frac{\sqrt{3}}{2} \hat \mu_N - \tilde \mu_N, (\psi - \psi_{f_N}) \res \clos(f_N)\rangle \right |\\
 & \le 2 |\tilde \mu_N|(\R^2) \
 r_\psi\Big( \frac 1 N\Big) \le C r_{\psi}\Big(\frac 1 N\Big) \to 0\,,
\end{aligned}
\end{equation}
where $r_\psi$ is the modulus of continuity of $\psi$\,. 

Now we prove that $|\frac 1 N\mu_N-\hat\mu_N|(\R^2)\to 0$ as $N\to +\infty$. 
We set $Z_N:=\{x\in\supp\mu_N\,:\,x\ \textrm {does not lie on four rhombic faces}\}$ and we claim that
\begin{equation}\label{not4}
 \sharp Z_N\leq {c}\sqrt N\,,
\end{equation}
for some positive constant $c\in\R$ (independent of $N$).

In order to prove \eqref{not4} we first notice that, in view of Lemma \ref{possibleanglesrmk}(iii),  $Z_N=Z_N^4{\cup} Z_N^{<4}$, where 
\begin{equation*}
\begin{aligned}
Z_N^4:=&\{x\in Z_N\,:\,x \textrm{ lies on four bonds}\}\,\textrm{ and}\\
Z_N^{<4}:=&\{x\in Z_N\,:\,x\textrm{ lies on $k$ bonds with }k\le 3 \}.
\end{aligned}
\end{equation*}
By \eqref{epdecompformula} and \eqref{enbou},  we have  
$$
\frac 1 4\sharp Z_N^4\le \sqrt N\Per^{\dg}(O_N(\mu_N))+\sharp\Ed^{\Wi}(\mu_N)+\sqrt N\sum_{f\in \Fbad(\mu_N)}\Per^\dg(f_N)\le
\bar C\sqrt{N},
$$
which, setting $c_1:=4\bar C$, implies that 
\begin{equation}\label{c1}
\sharp Z_N^4\le c_1\sqrt N.
\end{equation}
Moreover, again by \eqref{enbou} we get
$$ 
C\sqrt{N}\ge\E_N(\mu_N)+2\mu_N(\R^2)= \frac{1}{2} \sum_{x\in \supp\mu_N} \left( 4 -\sharp\{e\in\Ed(\mu_N)\ :\ x\in e\} \right)\geq \frac{1}{2}\sharp Z_N^{<4},
$$
which, setting $c_2:=2 C$, implies that
\begin{equation}\label{c2}
\sharp Z_N^{<4}\le c_2\sqrt N.
\end{equation}
 By summing \eqref{c1} and \eqref{c2} and setting $c:=c_1+c_2$, we get \eqref{not4}.
By \eqref{not4}, 
we get
\begin{equation}\label{vermenotil}
\begin{aligned}
\left|\frac 1 N\mu_N-\hat\mu_N\right|(\R^2)\leq \frac 1 N\sharp Z_N\le \frac {c}{\sqrt N}\to 0\,.
\end{aligned}
\end{equation}
By combining \eqref{cappmenotil} and \eqref{vermenotil} we obtain \eqref{fina}.
\end{proof}

\appendix

\section{Auxiliary lemmas} \label{aux2}
This appendix is devoted to the results used in the proof of Theorem \ref{fac}.

\begin{lemma}\label{lemma0}
Let $\bar N\in\N$\,. If \eqref{newii} holds true for every $N\le \bar N$\,, then every minimizer $(X,V)$ of $\E$ in $\AC_{\bar N+1}$ is such that $\G(X,V)$ is connected.
\end{lemma}
\begin{proof}
Let  $(X,V)$ be a minimizer of $\F$ in $\AC_{\bar N+1}$\,.
We argue by contradiction:
Let $\G_{\X_1},\ldots,\G_{\X_K}$ with $K\ge 2$ be the connected components of $\G(X,V)$ according to Definition \ref{conncomp} and set $\G_{\X_k}=\G(X_k, V_k)$ and $N_k:=\sharp\X_k$ for every $k=1,\ldots,K$\,. Then, $\bar N+1=\sum_{k=1}^{K} N_k$ and  $1\le N_k\le \bar N$ for every $k=1,\ldots,K$\,.

By the very definition of $\E$ in \eqref{defE0}
\begin{equation}\label{Edeco}
\E(X,V)=-\sum_{k=1}^K\sharp \Ed(X_k,V_k)=\sum_{k=1}^K\E(X_k,V_k)\,.
\end{equation}

By \eqref{Edeco} and
by 
the assumption \eqref{newii} for $N_k\le \bar N$, we
deduce
\begin{equation*}
\E(X,V)=\sum_{k=1}^{K}\E(X_k,V_k)
\ge \sum_{k=1}^{K}\E(Y_{N_k},W_{N_k}).
\end{equation*}
Now we will glue all the configurations $(Y_{N_k},W_{N_k})$ creating new bonds and hence decreasing the energy.
To this purpose, for every $k=1,\ldots,K$ and for every $\tau\in\R^2$ let $\Y_{N_k}(\tau):=\Y_{N_k}+\tau$ and let $(Y_{N_k}(\tau),W_{N_k})$ be the corresponding configuration. By construction, $\E(Y_{N_k},W_{N_k})=\E(Y_{N_k}(\tau),W_{N_k})$.
Since $W_{N_k}=\{(0;1)\}^{N_k}$ for every $k=1,\ldots,K$, it is easy to see that there exist translations $\tau_1,\ldots,\tau_k\in\R^2$ such that the configuration $(\bar Y_{\bar N+1},\bar W_{\bar N+1})$ with $\bar  W_{\bar N+1}=\{(0;1)\}^{\bar N+1}$ and $\bar\Y_{\bar N+1}:=\bigcup_{k=1}^K\Y_{N_k}(\tau_k)$ is in $\AC_{\bar N+1}$ and satisfies that $\G(\bar Y_{\bar N+1},\bar W_{\bar N+1})$ is connected (so that new bonds have been formed). 
It follows that
$$
\E(\bar Y_{\bar N+1},\bar W_{\bar N+1})<\sum_{k=1}^{K}\E(Y_{N_k},W_{N_k})\le\E(X,V),
$$
thus contradicting the minimality of $(X,V)$.
\end{proof}

Now we show that \eqref{newii} is satisfied for small number of particles.
\begin{lemma}\label{inspection}
For every $N\in\{1,\ldots,7\}$\,, \eqref{newii} holds true.
\end{lemma}
\begin{proof}
The claim is trivially satisfied for $N=1$ and $N=2$\,.
\vskip5pt
Note that for $N=3$\,, a minimizer $(X,V)$ of $\E$ in $\AC_N$ cannot have three bonds, since such three bonds should form an equilateral triangle with unitary side-length, which contradicts Lemma \ref{notri}. Therefore, we get
$$
-2\le \min_{(X,V)\in\AC_3}\E(X,V)\le \E(Y_3,W_3)=-2\,.
$$
\vskip5pt
Now we focus on the case $N=4$. Let $(X,V)$ be a minimizer of $\E$ in $\AC_4$. 
By Lemma \ref{lemma0} below $\G(X,V)$ is connected. Let $P$ be the maximum among the  perimeters $\Per_{\gr}(f)$ of the faces $f$ of $(X,V)$. If $P=0$, then $(X,V)$ has no faces and hence by Euler formula $\sharp\X-\sharp\Ed(X,V)=1$ so that
\begin{equation}\label{solofili}
\E(X,V)=-(\sharp\X-1)=-3> -4=\E(Y_4,W_4),
\end{equation}
thus contradicting the minimality of $(X,V)$.
It follows that $P\ge 3$ and in fact, by Lemma \ref{notri}, $P\ge 4$. Therefore, $P=4$ and by Lemma \ref{onlyrhombic}, $\X$ is the set of the vertices of a rhombus with unitary side-length and $v(x)$ is the same for every $x\in\X$. It follows that $(X,V)$ coincides - up to a rigid motion for $\X$ and up to flipping the orientation field $V$ - with $(Y_4,W_4)$.  

\vskip5pt
As for $N=5$,  let $(X,V)$ be a minimizer of $\E$ in $\AC_5$. 
By Lemma \ref{lemma0} $\G(X,V)$ is connected and 
by Lemmas \ref{notri} and \ref{onlyrhombic}, $\G(X,V)$ can have at most one face. Moreover, if $(X,V)$ has no faces, i.e., only wires are present, by arguing as in \eqref{solofili} we get a contradiction. 
As a consequence, $(X,V)$ has exactly one face $f$ and either $\Per(f)=5$ and there are no further edges or $\Per(f)=4$ and there is a wire edge.
In both cases we have
$$
\E(X,V)=-5=\E(Y_5,W_5)\,.
$$
\vskip5pt
Now we consider the case $N=6$. Let $(X,V)$ be a minimizer of $\E$ in $\AC_6$. By Lemma \ref{lemma0} $\G(X,V)$ is connected.
Since $(Y_6,W_6)$ is a competitor for the minimum problem in \eqref{newii}, we have
$$ 
\E(X,V)\le \E(Y_6,W_6)=-7\,.
$$
As in the previous cases, we can immediately exclude that $(X,V)$ has no faces. 
Moreover, $(X,V)$ cannot have only one face; indeed, if there is a unique face then by Euler formula
\begin{equation}\label{anchedopo}
\E(X,V)=-\sharp \X=-6>-7,
\end{equation}
thus contradicting the minimality of $(X,V)$.
In view of Lemma \ref{notri} we have that the only possibility is that $(X,V)$ has exactly two faces $f_1$ and $f_2$ with $\Per(f_j)=4$ ($j=1,2$); therefore, $(X,V)$ has no wires and, by Lemma \ref{onlyrhombic}, $f_1$ and $f_2$ are rhombuses with unitary side-length. Moreover, since $\G(X,V)$ is connected we have that $\partial f_1$ and $\partial f_2$ share at least one vertex
and actually, since $N=6$,  $f_1$ and $f_2$ share one bond. It follows that $(X,V)$ coincides - up to a rigid motion for $\X$ and up to flipping the orientation field $V$ - with $(Y_6,W_6)$.

\vskip5pt
We conclude by proving \eqref{newii} for $N=7$. 
Let $(X,V)$ be a minimizer of $\E$ in $\AC_7$. By Lemma \ref{lemma0} $\G(X,V)$ is connected. Since $(Y_7,W_7)$ is a competitor for the minimum problem in \eqref{newii}, we have
$$ 
\E(X,V)\le \E(Y_7,W_7)=-8\,.
$$
By arguing as in \eqref{solofili} and \eqref{anchedopo} we deduce that $(X,V)$ has at least two faces.

Notice that, by the hard sphere condition, it easily follows that if a face $f\in\Fa(X,V)$ has $\Per(f)\le 5$, then $f$ is convex and $\Per_\gr(f)=\Per(f)$.
If a face $f_1$ with maximal perimeter has $\Per(f_1)=5$ then, by Lemma \ref{notri}, the only possibility is that $(X,V)$ has no wires and one further face $f_2$ such that $\Per(f_2)=4$ and $f_1$ and $f_2$ share one edge. In such a case, we have that $\E(X,V)=-8=\E(Y_7,W_7)$. Now,  if a face $f_1$ with maximal perimeter has $\Per(f_1)=4$ then, by arguing as above we have that there exists  one further face $f_2$ such that $\Per(f_2)=4$. Since the graph is connected, we can have either that $f_1$ and $f_2$ share only one vertex and that there are no wires or that $f_1$ and $f_2$ share one edge and that there is a wire. In both cases we deduce that $\E(X,V)=-8= \E(Y_7,W_7)$ which implies the claim also for $N=7$ and concludes the proof of the lemma.
\end{proof}

By the very definition of $(Y_N,W_N)$ we have the following result.
\begin{lemma}\label{shellconstr}
Let $N,N'\in\N\cup\{0\}$ with $N>N'$\,. Let $(k;\delta)\in(\N\cup\{0\})\times (\N\cup\{0\})$ be the pair provided by \eqref{Ndecom} for $P=\Per_{\gr}(Y_{N'},W_{N'})$.
 Then 
\begin{equation}\label{form:shell}
\Per_{\gr}(Y_N,W_N)=\Per_{\gr}(Y_{N'},W_{N'})+8k+r_{N',N}\,,
\end{equation}
where $r_{N',N}\in \{0,2,4,6,8\}$ and $r_{N',N}=0$ if $\delta=0$.
Moreover, if $k\ge 1$\,, then $r_{N',N}\le 2\lceil\frac{\delta}{2}\rceil$.
\end{lemma}
\begin{proof}
We preliminarily note that $r_{N',N}$ is even since $\Per_{\gr}(Y_N,W_N)$ is even for every $N\in\N$\,. 
\vskip5pt

We first prove that if $\delta=0$, then \eqref{form:shell} holds true with $r_{N',N}=0$, i.e., 
we show that if $N=N'+k(\Per_{\gr}(Y_{N'},W_{N'})+4(k+1))$ for some $k\in\N$, then
\begin{equation}\label{barrata}
\Per_{\gr}(Y_{N},W_{N})=\Per_{\gr}(Y_{N'},W_{N'})+8k.
\end{equation}
We proceed by induction on $k$ starting from $k=1$.

Let $k=1$.
If $N'=0$ then $N=8$ and, by \eqref{percan}, we get $
\Per_{\gr}(Y_{N},W_{N})=8=\Per_{\gr}(Y_{ N'},W_{ N'})+8$.
If $N'\ge 1$, then $ N'=(l+1)^2+\eta$ for some $l,\eta\in\N\cup\{0\}$ with 
$0\le \eta\le 2l+2$.
We prove \eqref{barrata} (with $k=1$) only for $\eta=0$, the other cases being fully analogous.
Since $\eta=0$, we have $N'=(l+1)^2$ for some $l\in\N\cup\{0\}$; then, by \eqref{percan}, $\Per_{\gr}(Y_{N'},W_{N'})=4l$ so that
$$
N=N'+\Per_{\gr}(Y_{N'},W_{N'})+8=(l+1)^2+4l+8=(l+3)^2\,,
$$
which, in view of \eqref{percan} yields 
$$
\Per_{\gr}(Y_{N},W_{N})=4(l+2)=4l+8=\Per_{\gr}(Y_{N'},W_{N'})+8,
$$
i.e., \eqref{barrata}. 


\vskip5pt

Assume now that 
\eqref{barrata} holds true for every $k= 1,\ldots, \hat k$, 
and let us show that it is satisfied also for $k= \hat k+1$.
Let $N'\in\N\cup\{0\}$, $k=\hat k+1$ and, in turn, $N=N'+(\hat k+1)(\Per_{\gr}(Y_{N'},W_{N'})+4(\hat k+2))$. We set $\bar N':=N'+\hat k(\Per_{\gr}(Y_{N'},W_{N'})+4(\hat k+1))$; applying the inductive step to $N'$, $k=\hat k$ and $N=\bar N'$ we have 
\begin{equation}\label{sper}
\Per_{\gr}(Y_{\bar N'},W_{\bar N'})=\Per_{\gr}(Y_{N'},W_{N'})+8\hat k.
\end{equation}
It follows that
\begin{align*}
N=&N'+(\hat k+1)(\Per_{\gr}(Y_{N'},W_{N'})+4(\hat k+2)) \\
=&N'+\hat k(\Per_{\gr}(Y_{N'},W_{N'})+4(\hat k+1))+4\hat k+\Per_{\gr}(Y_{N'},W_{N'})+4(\hat k+2)\\
=&\bar N'+\Per_{\gr}(Y_{N'},W_{N'})+8\hat k+8=\bar N'+\Per_{\gr}(Y_{\bar N'},W_{\bar N'})+8.
\end{align*}
By the step $k=1$ applied to $N'$ replaced by $\bar N'$ and, in turn, $N=N$ and by \eqref{sper} we get 
$$
\Per_{\gr}(Y_N,W_N)=\Per_{\gr}(Y_{\bar N'},W_{\bar N'})+8=\Per_{\gr}(Y_{N'},W_{N'})+8(\hat k+1)\,,
$$
thus implying the claim and concluding the proof of \eqref{form:shell} for $\delta=0$ with $r_{N',N}=0$.
\vskip5pt
Now we prove \eqref{form:shell} in the general case. 
Let $N=N'+k(\Per_{\gr}(Y_{N'},W_{N'})+4(k+1))+\delta$ with $\delta<\Per_{\gr}(Y_{N'},W_{N'})+8(k+1)$.
We set $N_1:=N'+k(\Per_{\gr}(Y_{N'},W_{N'})+4(k+1))$ and $N_2:=N'+(k+1)(\Per_{\gr}(Y_{N'},W_{N'})+4(k+2))$ we have that $N_1\le N< N_2$;
by applying \eqref{barrata}  with $N$ replaced by $N_1$ and $N_2$ respectively and using the monotonicity of $\Per_{\gr}(Y_{N},W_{N})$ with respect to $N$ we obtain
\begin{align*}
\Per_{\gr}(Y_{N'},W_{N'})+8k=&\Per_{\gr}(Y_{N_1},W_{N_1})\\
\le& \Per_{\gr}(Y_{N},W_{N})\le \Per_{\gr}(Y_{N_2},W_{N_2})=\Per_{\gr}(Y_{N'},W_{N'})+8(k+1),
\end{align*}
which concludes the proof of \eqref{form:shell}.
\vskip5pt
Finally we show that $r_{N',N}\le 2\lceil\frac\delta 2\rceil$ if $k\ge 1$. 
Let $N'\in\N\cup\{0\}$ and let $k\ge 1$. 
Let $N:=N'+k(\Per_{\gr}(Y_{N'},W_{N'})+4(k+1))+\delta$ 
and set 
\begin{equation*}
\tilde N:=N'+k(\Per_{\gr}(Y_{N'},W_{N'})+4(k+1))
\end{equation*}
so that $N=\tilde N+\delta$.
Since we are assuming that $k\ge 1$, we have that $\tilde N\ge 8$. 
By \eqref{form:shell} and by applying \eqref{barrata} with $N'=N'$, $k=k$ and $N=\tilde N$ we have
\begin{equation}\label{percosa}
\Per_{\gr}(Y_N,W_N)=\Per_{\gr}(Y_{N'}, W_{N'})+8k+r_{N',N}=\Per_{\gr}(Y_{\tilde N}, W_{\tilde N})+r_{N',N}.
\end{equation}
For every $\delta\in\N\cup\{0\}$ we set 
\begin{equation}\label{percosa2}
p(\delta):=\Per_{\gr}(Y_{\tilde N+\delta},W_{\tilde N+\delta})-\Per_{\gr}(Y_{\tilde N},W_{\tilde N}).
\end{equation}
Note that $p(\delta)\in 2\N\cup\{0\}$ for every $\delta\in\N\cup\{0\}$ and $p(0)=0$. We denote by $J_{p}:=\{\delta_i\}\subseteq\N$ with $\delta_i<\delta_{i+1}$ the set of the ``jump discontinuities'' of $p$, i.e., such that $p(\delta_i)>p(\delta_i-1)$. Since $\tilde N\ge 8$, in view of \eqref{percan}, for every $\delta_i\in J_p$ we have
$$
p(\delta_i)-p(\delta_i-1)=2,\qquad
|\delta_{i+1}-\delta_{i}|\ge 3.
$$ 
This fact together with $p(0)=0$ yields $p(\delta)\le 2\lceil\frac\delta 2\rceil$ for every $\delta\in\N$ and hence in view of \eqref{percosa} and \eqref{percosa2} we obtain that $r_{N',N}=p(\delta)\le 2\lceil\frac\delta 2\rceil$ for every $\delta\in\N$.

%
%
%
%

\end{proof}
\begin{lemma} \label{puntidibordo}
Let $N\in\N$ with $N\ge 2$ and let $(X,V)\in\AC_N$ be such that $\G(X,V)$ is connected. 
Then
\begin{equation}\label{per>bor}
\Per_{\gr}(X,V)
\ge \sharp\partial X.
\end{equation}
%
%
\end{lemma}
\begin{proof}
Set $\tilde \X:=\partial X$, $\tilde \Ed:=\Ed^{\partial}(X,V)\cup\Ed^{\Wi,\ext}(X,V)$ and $\tilde\G=(\tilde\X,\tilde\Ed)$. According to Section \ref{prelimgra},  $\Fa(\tilde\G)$ denotes the set of faces of $\tilde\G$. 

Since $\G(X,V)$ is connected, we have that $\tilde\G$ is connected too, so that by Lemma \ref{lm:bouchar} we have
$$
\sharp \tilde \X-\sharp\tilde \Ed+\sharp\Fa(\tilde\G)=1,
$$
i.e.,
\begin{align}\label{base}
\sharp \partial X=1-\sharp\Fa(\tilde\G)+\sharp\Ed^{\partial}(X,V)+\sharp \Ed^{\Wi,\ext}(X,V).
\end{align}
If $\sharp\Fa(\tilde\G)=0$, then $\sharp\Ed^{\partial}(X,V)=0$ and, since $N\ge 2$ and  $\G(X,V)$ is connected, we have $\sharp \Ed^{\Wi,\ext}(X,V)\ge 1$, so that
$$
\sharp \partial X=1+\sharp \Ed^{\Wi,\ext}(X,V)\le 2\sharp \Ed^{\Wi,\ext}(X,V)=\Per_{\gr}(X,V).
$$
If $\sharp\Fa(\tilde\G)\ge 1$, by \eqref{base} we get
$$
\sharp \partial X\le\sharp\Ed^{\partial}(X,V)+\sharp \Ed^{\Wi,\ext}(X,V)\le \sharp\Ed^{\partial}(X,V)+2\sharp \Ed^{\Wi,\ext}(X,V)=\Per_{\gr}(X,V).
$$
In both cases we have proven \eqref{per>bor}.
\end{proof}

\begin{lemma}\label{increase}
Let $N\in\N$ with $N\ge 8$\,.
Let  $(X,V)\in\AC_N$ with $\G(X,V)$ connected.
Then
\begin{equation*}
\F(X,V)\ge \F(X',V')+4\,,
\end{equation*}
where $(X', V')$ is as in Definition \ref{borX}.
\end{lemma}
Before proving Lemma \ref{increase},
we need some auxiliary lemmas.
We start by fixing some notations.

Let $N\ge 8$ and let $(X,V)\in\AC_N$ be such that $\G(X,V)$ is connected, $O(X,V)$ has simple and closed polygonal boundary and $\sharp\Ed^{\Wi,\ext}(X,V)=0$\,.
For every $x\in\partial X$, we set
\begin{equation}\label{idef}
\ii(x):=\sharp\{e\in\Ed^{\inte}(X,V)\,:\,x\in e\}.
\end{equation}
Moreover, for every $x\in \partial X$ we denote by $\hat x$ the inner angle spanned by the two boundary edges containing $x$. Here by inner angle we mean the angle that is ``interior'' to $O(X,V)$.
In view of Lemma \ref{possibleanglesrmk} we can classify the points $x\in\partial X$ in the following subclasses:
\begin{equation*}
\begin{aligned}
Y^{j}:=\Big\{x\in\partial X\,:\, \hat x=j\frac \pi 3\Big\},&\qquad j=1,\ldots,5\,,\\
Y^{j,j+1}:=\Big\{x\in\partial X\,:\, \hat x\in\Big(j\frac\pi 3,(j+1)\frac \pi 3\Big)\Big\},&\qquad j=2,3\,,
\end{aligned}
\end{equation*}
i.e.,
\begin{equation}\label{classification}
\partial X=\bigcup_{j=1}^{5} Y^j\cup Y^{2,3}\cup Y^{3,4}.
\end{equation}
By Lemma \ref{possibleanglesrmk}(iii) every point in $\X$ can lie on at most four bonds, so that $\ii(x)\le 2$ for every $x\in\partial X$\,.
Therefore, for every $k=0,1,2$ we can set 
\begin{equation*}
\begin{aligned}
Y_k^j:=&\{ x\in Y^j\,:\, \ii(x)=k \}\qquad\textrm{ for }j=1,2,3,4,5\,,\\
Y_k^{j,j+1}:=&\{ x\in Y^{j,j+1}\,:\, \ii(x)=k \}\qquad\textrm{ for }j=2,3\,. 
\end{aligned}
\end{equation*}
By Lemma \ref{possibleanglesrmk} we have
\begin{multline}\label{classY}
Y^1=Y^1_0\,,\quad Y^2=Y^2_0\,,\quad Y^{2,3}=Y^{2,3}_0\,,\quad Y^3=Y_0^3\cup Y_1^3\,, \\
\quad Y^{3,4}=Y_{0}^{3,4}\cup Y_1^{3,4}\,,\quad Y^4=Y_0^4\cup Y_1^4\cup Y_2^4\,,\quad Y^5=Y_0^5\cup Y_1^5\cup Y_2^5\,. 
\end{multline}
\begin{definition}\label{consecutivedef}
Given two boundary particles $x',x''\in Y^1\cup Y^2\cup Y^{2,3}\cup Y^4\cup Y^5=\partial X\setminus(Y^3\cup Y^{3,4})$ we say that $x'$ \textit{follows} $x''$ and we write $x'\to x''$ if there exist $M\in\N$ and a path $x'=z_0,\dots,z_M=x''$ in  $\partial X$,
oriented according to the counter-clockwise orientation of $\partial O(X,V)$, such that $\{z_{m-1},z_{m}\}\in\Ed^{\partial}(X,V)$ 
and
 $z_m\in Y^3\cup Y^{3,4} $ for every $m=1,\dots,M-1$. 
 \end{definition}
 Note that the case $M=1$ in Definition \ref{consecutivedef} corresponds to $\{x',x''\}\in\Ed^{\partial}(X,V)$.
Notice also that for every point $x'\in\partial X\setminus(Y^3\cup Y^{3,4})$ there exists a unique point $x''\in\partial X\setminus(Y^3\cup Y^{3,4})$ such that $x'\to x''$.
Moreover, we set
\begin{equation}\label{consecutive}
Y_k^j \to Y_a^b  :=\{ x'\in Y_k^j\,:\,x'\to x'' \textrm{ for some } x''\in Y_a^b  \}\,, 
\end{equation}
and $Y^j_k\to \{Y_{a_1}^{b_1},\ldots,Y_{a_L}^{b_L}\}:=\bigcup_{l=1}^{L}Y_k^j \to Y_{a_l}^{b_l} $\,.

For every $x\in\partial X$ and for every $z_1,z_2\in \X$ with $\{x,z_1\}\,\{x,z_2\}\in\Ed(X,V)$\,, we denote by $\widehat{z_1xz_2}$ the angle formed by $[x,z_1]$ and $[x,z_2]$ that is ``interior'' to $O(X,V)$\,. Notice that, if also $z_1,z_2\in\partial X$, then $\{x,z_1\}\,\{x,z_2\}\in\Ed^\partial(X,V)$ and $\hat x=\widehat{z_1xz_2}$\,. 
 

As a consequence of Lemma \ref{possibleanglesrmk}, we have the following result.

\begin{lemma}\label{A.1}

Let $(X,V)\in\AC$ be such that $\G(X,V)$ is connected, $O(X,V)$ has simple and closed polygonal boundary and $\Ed^{\Wi,\ext}(X,V)=\emptyset$. The following facts hold true.

\begin{enumerate}
\item If $x,y,z_1,z_2\in\X$ are such that $x\in Y^3_1\cup Y^{3,4}_1$\,, $\{x,z_1\},\{x,z_2\}\in\Ed^{\partial}(X,V)$ and $\{x,y\}\in\Ed^{\inte}(X,V)$, then either $\widehat{z_1xy}=\frac\pi 3$  or   $\widehat{yxz_2}=\frac \pi 3$\,. Moreover, $v(x)$ is parallel to  the bisector of the $\frac\pi 3$ angle.
\item If $x,y_1,y_2,z_1,z_2\in\X$ are such that $x\in Y^4_2$\,, $\{x,z_1\},\{x,z_2\}\in\Ed^{\partial}(X,V)$\,,\\
 $\{x,y_1\}\,,\{x,y_2\}\in\Ed^{\inte}(X,V)$ and $z_1,z_2,y_2,y_1$ are counter-clockwise oriented, then $\widehat{z_1xy_1}=\frac\pi 3$\,, $\widehat{y_1xy_2}=\frac 2 3\pi$ and  $\widehat{y_2xz_2}=\frac \pi 3$\,.
Moreover, $v(x)$ is parallel to the bisector of the angles $\widehat{z_1xy_1}$ and $\widehat{y_2xz_2}$\,.
\item If $x,y_1,y_2,z_1,z_2\in\X$ are such that $x\in Y^5_2$\,, $\{x,z_1\},\{x,z_2\}\in\Ed^{\partial}(X,V)$\,,\\
 $\{x,y_1\}\,,\{x,y_2\}\in\Ed^{\inte}(X,V)$ and and $z_1,z_2,y_2,y_1$ are counter-clockwise oriented, then $\widehat{z_1xy_1}=\frac 2 3\pi $\,, $\widehat{y_1xy_2}=\frac \pi 3$ and  $\widehat{y_2xz_2}=\frac 2 3\pi$\,. Moreover, $v(x)$ is parallel to the bisector of the angle $\widehat{y_1xy_2}$\,.
%
%
\end{enumerate}
\end{lemma}

\begin{lemma}\label{A.2}
Let $(X,V)\in\AC$ be such that $\G(X,V)$ is connected, $O(X,V)$ has simple and closed polygonal boundary and $\Ed^{\Wi,\ext}(X,V)=\emptyset$. 
 Let $x,y\in \partial X$.
 
 If $x,y\in Y_0^1$ and $x\to y$, then 
 $M\ge 2$ in Definition \ref{consecutivedef} and there exists $z_m$ as in Definition  \ref{consecutivedef} such that  $z_m\in Y_0^3\cup Y_0^{3,4}$.

The same statement holds true also if $x,y\in Y_2^4$ or if $x\in Y_0^1$\,, $y\in Y_2^4$ or if $x\in Y_2^4$\,, $y\in Y_0^1$.

\end{lemma}
\begin{proof}
We prove the claim only for $x,y\in Y_0^1$, being the proof in the other cases fully analogous.
By Lemma \ref{notri}, we immediately get
$M\ge 2$.
 Assume by contradiction that $z_m\in Y_1^3\cup Y_1^{3,4}$ for every $m=1,\ldots,M-1$\,. For every $m=1,\ldots,M-1$ let $w_m$ be the only point in $\X$ such that $\{z_m,w_m\}\in\Ed^{\inte}(X,V)$\,. By Lemma \ref{notri}  we have that the points $w_m$ are all distinct. Moreover, by  Lemma \ref{notri} and by Lemma \ref{A.1}(i) we get $\widehat{w_1z_1z_{2}}=\frac{\pi}{3}$; analogously, it follows that $\widehat{w_2z_2z_{3}}=\frac{\pi}{3}$
 and in fact, by induction, that  $\widehat{w_mz_mz_{m+1}}=\frac{\pi}{3}$ for every $m=1,\ldots, M-1$\,.
In particular, $\widehat{w_{M-1}z_{M-1}z_{M}}=\frac{\pi}{3}$, thus contradicting Lemma \ref{notri}.

\end{proof}
\begin{proof}[Proof of Lemma \ref{increase}]
We preliminarily prove the claim under the assumptions that 
$O(X,V)$ has simple and closed polygonal boundary and $\sharp\Ed^{\Wi,\ext}(X,V)=0$\,;
in such a case
\begin{equation}\label{nowires}
\Per_{\gr}(X,V)=\Per(X,V)=\sharp\partial X\,.
\end{equation} 
Set 
$$
I:=\{e\in\Ed^{\inte}(X,V)\,:\, \exists x\in \partial X \textrm{ such that }x\in e\};
$$
recalling \eqref{idef}, we have
\begin{equation}\label{iI}
\sharp I\le \sum_{x\in\partial X}\ii(x).
\end{equation}
In view of \eqref{defE0}, \eqref{decompformula}, and \eqref{surface+}, we have
\begin{equation*}
\begin{aligned}
-\Per(X,V)-\sharp I=\E(X,V)-\E(X',V')=&-2\sharp\X+\F(X,V)+2\sharp\X'-\F(X',V')\\
=&-2\sharp\partial X+\F(X,V)-\F(X',V'),
\end{aligned}
\end{equation*}
whence, together with \eqref{nowires} and \eqref{iI},  we deduce 
\begin{equation*}
\F(X,V)-\F(X',V')=\sharp \partial X-\sharp I\ge\sum_{x\in\partial X}(1-\ii(x))\,.
\end{equation*}
Therefore, in order to prove the claim it is enough to show that
\begin{equation}\label{ineqnew}
\sum_{x\in\partial X}(1-\ii(x))\ge 4\,.
\end{equation}
Recalling \eqref{classification} and \eqref{classY} we have
\begin{equation}\label{starting}
\begin{aligned}
\sum_{x\in\partial X}(1-\ii(x))=&\sharp Y^1_0-\sharp Y^5_2+\sharp Y^2_0-\sharp Y^4_2\\
&+\sharp Y^{2,3}_0+\sharp Y^3_0+\sharp Y^{3,4}_0+\sharp Y^4_0+\sharp Y^5_0\,.
\end{aligned}
\end{equation}
Since  $O(X,V)$ is connected, the Gauss-Bonnet formula gives $2\pi=\sum_{x\in\partial X}(\pi-\hat x)$\,, where $\hat x$ denotes the angle at $x$ interior to $O(X,V)$;
therefore,
\begin{equation*}
\begin{aligned}
2\pi=&\sum_{x\in\partial X}(\pi-\hat x)\le \frac{2}{3}\pi \sharp Y^1+\frac{\pi}{3} \sharp Y^2+\frac{\pi}{3} \sharp Y^{2,3}-\frac{\pi}{3}\sharp Y^{4}-\frac{2}{3}\pi\sharp Y^5,
\end{aligned}
\end{equation*}
which, in view of \eqref{classY}, yields
\begin{equation}\label{gb1}
\begin{aligned}
6\le& 2(\sharp Y^1_0-\sharp Y^5_2)+\sharp Y_0^2-\sharp Y^4_2-r
\end{aligned}
\end{equation}
where 
\begin{equation}\label{defr}
r:=2\sharp Y^5_1+2\sharp Y^5_0+\sharp Y_1^4+\sharp Y^4_0-\sharp Y^{2,3}_0\,.
\end{equation}
We point out that we are making no claims on the  on the sign of $r$.



Recalling the notation introduced in \eqref{consecutive},
for every $(j;k)\in\{(1;0),(2;0),(4;2),(5;2)\}$, we set
\begin{equation*}
\begin{aligned}
A_k^j:=Y_k^j\to \{Y_0^2, Y^5_2\},\quad
B_k^j:=Y_k^j\to \{Y_0^1, Y^4_2\},\quad
C_k^j:=Y_k^j\to \{Y_0^{2,3},Y_0^4, Y_1^4, Y_0^5, Y_1^5\}\,,
\end{aligned}
\end{equation*}
so that 
\begin{equation}\label{classABC}
\sharp Y_k^j= \sharp A_k^j + \sharp  B_k^j +\sharp C_k^j\qquad\textrm{for every }(j;k)\in\{(1;0),(2;0),(4;2),(5;2)\}.
\end{equation}
By construction
\begin{equation}\label{Y2052}
\sharp A_0^1+\sharp A_0^2+ \sharp A_2^4 +\sharp A_2^5\le \sharp Y_0^2+\sharp Y^5_2\,,
\end{equation}
\begin{equation}\label{Csum}
\sharp C^1_0+\sharp C^2_0+ \sharp C^4_2+\sharp C^5_2\le \sharp Y_0^{2,3}+\sharp Y_0^4+ \sharp Y_1^4+\sharp Y_0^5+ \sharp Y^5_1\,,
\end{equation}
and, by Lemma \ref{A.2},
\begin{equation}\label{paying}
\sharp B_0^1+\sharp B_2^4\le \sharp Y_0^3+\sharp Y_0^{3,4}\,.
\end{equation}
By applying \eqref{classABC} with $(j;k)=(1;0)$ and $(j;k)=(4;2)$, summing such identities, and using  \eqref{Y2052}, we have
\begin{equation}\label{ineq0}
\begin{aligned}
\sharp Y_0^1+\sharp Y_2^4-(\sharp B_0^1+\sharp C_0^1+\sharp B_2^4+\sharp C_2^4)=&\sharp A_0^1+\sharp A_2^4\\
\le&\sharp A_0^1+\sharp A_2^4+ \sharp A_0^2 +\sharp A_2^5 \le \sharp Y_0^2+\sharp Y_2^5.
\end{aligned}
\end{equation}
Moreover, 
in view of \eqref{Csum} and \eqref{paying}, we have
\begin{equation}\label{bardelta}
\begin{aligned}
&-(\sharp B_0^1+\sharp C_0^1+\sharp B_2^4+\sharp C_2^4)\\
\ge& -(\sharp B_0^1+\sharp B_2^4+\sharp C^1_0+\sharp C^2_0+ \sharp C^4_2+\sharp C^5_2)\\
\ge&-\sharp Y_0^{2,3}-\sharp Y_0^4- \sharp Y_1^4-\sharp Y_0^5-\sharp Y^5_1- \sharp Y_0^3-\sharp Y_0^{3,4}=:\bar\delta\,,
\end{aligned}
\end{equation}
which, by \eqref{ineq0}, yields 
\begin{equation*}
\sharp Y_0^1+\sharp Y_2^4 +\bar\delta \le \sharp Y_0^2+\sharp Y_2^5\,.
\end{equation*}
Therefore
\begin{equation*}
\sharp Y_0^2+\sharp Y_2^5=\sharp Y_0^1+\sharp Y_2^4 +\delta\qquad\textrm{for some }\delta\ge\bar\delta\,,
\end{equation*}
or, equivalently,
\begin{equation}\label{equequ}
\sharp Y_0^2-\sharp Y_2^4=\sharp Y_0^1-\sharp Y_2^5 +\delta\qquad\textrm{for some }\delta\ge\bar\delta\,.
\end{equation}
By \eqref{equequ} and by \eqref{gb1}, we deduce that
\begin{equation*}
6\le 2(\sharp Y_0^1-\sharp Y_2^5)+\sharp Y_0^1-\sharp Y_2^5+\delta-r
\end{equation*}
so that
\begin{equation}\label{comb}
\sharp Y_0^1-\sharp Y_2^5 \ge 2+\frac{r}{3}-\frac{\delta}{3}\,,
\end{equation}
and hence, by \eqref{equequ}, we get
\begin{equation}\label{comb2}
\quad \sharp Y_0^2-\sharp Y_2^4\ge 2+\frac{r}{3}+\frac{2}{3}\delta\,.
\end{equation}
By \eqref{starting}, \eqref{comb}, \eqref{comb2}, \eqref{defr} and \eqref{bardelta} we obtain
\begin{equation*}
\begin{aligned}
\sum_{x\in\partial X}(1-\ii(x))\ge &4+\frac{2}{3}r+\frac{\bar\delta}{3}+\sharp Y^{2,3}_0+\sharp Y^3_0+\sharp Y^{3,4}_0+\sharp Y^4_0+\sharp Y^5_0\\
\ge & 4+\frac 4 3\sharp Y^5_1+\frac 4 3\sharp Y^5_0+\frac 2 3\sharp Y_1^4+\frac 2 3\sharp Y^4_0-\frac 2 3\sharp Y^{2,3}_0\\
&-\frac 1 3\sharp Y_0^{2,3}-\frac 1 3\sharp Y_0^4-\frac 1 3 \sharp Y_1^4-\frac 1 3\sharp Y_0^5-\frac 1 3\sharp Y^5_1-\frac 1 3 \sharp Y_0^3-\frac 1 3\sharp Y_0^{3,4}\\
&+\sharp Y^{2,3}_0+\sharp Y^3_0+\sharp Y^{3,4}_0+\sharp Y^4_0+\sharp Y^5_0\\
=&4+\frac{2}{3}\sharp Y_0^3+\frac{2}{3} \sharp Y_0^{3,4}+\frac{4}{3}\sharp Y_0^4+2\sharp Y_0^5+\frac 1 3 \sharp Y_1^4+\sharp Y_1^5\\
\ge &4\,,
\end{aligned}
\end{equation*}
which implies \eqref{ineqnew} and hence the claim in the case that $O(X,V)$ has simple and closed polygonal boundary and $\sharp\Ed^{\Wi,\ext}(X,V)=0$\,.

\vskip5pt

If  $O(X,V)$ is connected and  $\Ed^{\Wi,\ext}(X,V)=\emptyset$ (without assuming that $O(X,V)$ has simple and closed polygonal boundary), we argue in the following manner.
By Lemma \ref{possibleanglesrmk}(iii), every point $x\in\partial X$ can lie either on two boundary edges or on four boundary edges. 
We set
$$
\X_{\bow}:=\{x\in\partial X\,:\, x\textrm{ lies on four boundary edges}\}.
$$
Let $\hat O_1(X,V),\ldots,\hat O_K(X,V)$ be the connected components of $O(X,V)\setminus\X_{\bow}$. 
For every $k=1,\ldots,K$  we set $\tilde O_k(X,V):=\clos(\hat O_k(X,V))$, $\tilde \X_k:=\X\cap \tilde O_k(X,V)$ and we denote by  $(\tilde X_k,\tilde V_k)$ the corresponding configuration. Clearly, $O(\tilde X_k,\tilde V_k)=\tilde O_k(X,V)$ has simple and closed polygonal boundary, $\G(\tilde X_k,\tilde V_k)$ is connected and $\Ed^{\Wi,\ext}(\tilde X_k,\tilde V_k)=\emptyset$; therefore, by the above proven result, we get
\begin{equation}\label{comeprima}
\F(\tilde X_k,\tilde V_k)\ge \F(\tilde X'_k,\tilde V'_k)+4\qquad\textrm{for every }k=1,\ldots,K.
\end{equation}
Since $\sum_{k=1}^K\F(\tilde X'_k,\tilde V'_k)=\F(X',V')$, by \eqref{comeprima} we get
\begin{equation}\label{bows}
\begin{aligned}
\F(X,V)=&\frac 1 2 \Per_{\gr}(X,V)+\frac 1 2 \defe_\gr(X,V)+2\\
=&\sum_{k=1}^K\Big(\frac 1 2 \Per_{\gr}(\tilde X_k,\tilde V_k)+\frac 1 2 \defe_\gr(\tilde X_k,\tilde V_k)\Big)+2\\
=&\sum_{k=1}^K\F(\tilde X_k,\tilde V_k)+2(1-K)
\ge 
\sum_{k=1}^K \Big(\F(\tilde X'_k,\tilde V'_k)+4\Big)+2(1-K)\\
=& \F(X',V')+2K+2
\ge\F(X',V')+4.
\end{aligned}
\end{equation}

\vskip5pt

Finally we treat the general case.
Let $O_1(X,V),\ldots,O_K(X,V)$ be the connected components of $O(X,V)$. 
For every $k=1,\ldots,K$  we set $\X_k:=\X\cap O_k(X,V)$ and we denote by  $(X_k,V_k)$ the corresponding configuration. Clearly, $O(X_k,V_k)=O_k(X,V)$ and $\Ed^{\Wi,\ext}(X_k, V_k)=\emptyset$, so that, by \eqref{bows}, we have
\begin{equation}\label{comeprima2}
\F(X_k,V_k)\ge \F(X'_k,V'_k)+4\qquad\textrm{for every }k=1,\ldots,K.
\end{equation}
Now, since $\sum_{k=1}^K\F(X'_k,V'_k)=\F(X',V')$, by \eqref{comeprima2} we obtain
\begin{equation*}
\begin{aligned}
\F(X,V)=&\frac 1 2 \Per_{\gr}(X,V)+\frac 1 2 \defe_\gr(X,V)+2\ge\sum_{k=1}^K\Big(\frac 1 2 \Per_{\gr}(X_k,V_k)+\frac 1 2 \defe_\gr(X_k,V_k)\Big)+2\\
=&\sum_{k=1}^K\F(X_k,V_k)+2(1-K)
\ge 
\sum_{k=1}^K \Big(\F(X'_k,V'_k)+4\Big)+2(1-K)\ge\F(X',V')+4,
\end{aligned}
\end{equation*}
which concludes the proof of the Lemma.

\end{proof}

\end{document}